\documentclass[10pt,journal,final,twocolumn]{IEEEtran}
\IEEEoverridecommandlockouts
\usepackage{amsmath,amssymb,amsfonts}
\usepackage{cite}
\usepackage{graphicx}
\usepackage{epstopdf}

\usepackage{booktabs}
\usepackage{graphicx}
\usepackage{textcomp}
\usepackage{xcolor}
\usepackage{times}
\usepackage{subfigure}
\usepackage{amssymb,amsmath}
\usepackage{acronym}  % make an acronym
\usepackage{balance}

\usepackage{bm}         
\usepackage{algorithm} 
\usepackage{algorithmic}

\usepackage{lipsum}
\usepackage{stfloats}

\usepackage{url}
\usepackage{mathrsfs}
\usepackage{bbm}
\usepackage{amsthm}
\usepackage{enumerate}

\usepackage{hyperref}
\hypersetup{hidelinks, 
colorlinks=true,
allcolors=black,
pdfstartview=Fit,
breaklinks=true}

\theoremstyle{definition}

\newtheorem{lemma}{\bf Lemma}
\newtheorem{corollary}{\bf Corollary}

\theoremstyle{remark}
\newtheorem{remark}{\bf Remark}

\acrodef{ofdm}[OFDM]{orthogonal frequency division multiplexing}%
\acrodef{miso-ofdm}[MISO-OFDM]{multi-input single-output orthogonal frequency division multiplexing}%

\acrodef{ris}[RIS]{reconfigurable intelligent surface}%
\acrodef{qos}[QoS]{quality of service}%
\acrodef{idft}[IDFT]{inverse discrete Fourier transform}%
\acrodef{dft}[DFT]{discrete Fourier transform}%
\acrodef{cp}[CP]{cyclic prefix}%
\acrodef{csi}[CSI]{channel state information}%
\acrodef{awgn}[AWGN]{additive white Gaussion noise}%

\acrodef{qcqp}[QCQP]{quadratically constrained quadratic program}%
\acrodef{qp}[QP]{quadratic program}%
\acrodef{bs}[BS]{base station}%
\acrodef{ap}[BS]{base station}%
\acrodef{aps}[APs]{access points}%
\acrodef{qos}[QoS]{quality of service}%
\acrodef{ue}[UE]{user equipment}%
\acrodef{snr}[SNR]{signal-to-noise ratio}%
\acrodef{mmwave}[mmWave]{millimeter-wave}%
\acrodef{snr}[SNR]{signal-to-noise ratio}%

\acrodef{sinr}[SINR]{signal-to-interference-plus-noise ratio}%

\acrodef{ser}[SER]{symbol error rate}%
\acrodef{rc}[RC]{reflection coefficient}%
\acrodef{uavs}[UAVs]{unmanned aerial vehicles}%
\acrodef{mimo}[MIMO]{multiple-input multiple-output}%
\acrodef{noma}[NOMA]{non-orthogonal multiple access}%

\acrodef{ace}[ACE]{adaptive cross-entropy}%
\acrodef{wsr}[WSR]{weighted sum-rate}%
\acrodef{udn}[UDN]{ultra-dense network}%
\acrodef{Udn}[UDN]{Ultra-dense network}%

\def\BibTeX{{\rm B\kern-.05em{\sc i\kern-.025em b}\kern-.08em
    T\kern-.1667em\lower.7ex\hbox{E}\kern-.125emX}}

\begin{document}
\title{The Manifestation of Spatial Wideband Effect in Circular Array: From Beam Split to Beam Defocus} %% Delay-phase Precoding to Alleviate Wideband Beam Misfocusing for Circular Arrays
%% From Beam Split for Linear Arrays to Beam Defocus for Circular Arrays
\author{\IEEEauthorblockN{Zidong Wu, \emph{Student Member, IEEE}, and Linglong Dai, \emph{Fellow, IEEE}}

\thanks{All authors are with the Department of Electronic Engineering, Tsinghua University as well as Beijing National Research Center for Information Science and Technology (BNRist), Beijing 100084, China (e-mails: wuzd19@mails.tsinghua.edu.cn, daill@tsinghua.edu.cn).}
}

\maketitle
\vspace{-4em}
\begin{abstract}
Millimeter-wave (mmWave) and terahertz (THz) communications with hybrid precoding architectures have been regarded as energy-efficient solutions to fulfill the vision of high-speed transmissions for 6G communications. Benefiting from the advantages of providing a wide scan range and flat array gain, the uniform circular array (UCA) has attracted much attention. However, the growing bandwidth of mmWave and THz communications require frequency-independent phase shifts, which can not be perfectly realized through frequency-independent phase shifters (PSs) in classical hybrid precoding architectures. This mismatch causes the beam defocus effect in UCA wideband communications, where the high-gain beams could not form at non-central frequencies. In this paper, we first investigate the characteristics of the beam defocus effect distinguishing itself from the beam split effect in uniform linear array (ULA) systems. The beam pattern of UCA in both frequency domain and angular domain is analyzed, characterizing the beamforming loss caused by the beam defocus effect. Then, the delay-phase-precoding (DPP) architecture which leverages the true-time-delay (TTD) devices to generate frequency-dependent phase shifts is employed to mitigate the beam defocus effect. Finally, performance analysis and extensive simulation results are provided to evaluate the effectiveness of the DPP architecture in UCA systems.
\end{abstract}

\section{Introduction}\label{sec: intro}
\par To meet the ever-increasing growth of data transmission demand, high-frequency bands such as millimeter-wave (mmWave) and terahertz (THz) are promising to provide abundant spectrum resources for future sixth-generation (6G) communications. Nevertheless, high-frequency communications have to face the challenge of severe propagation attenuation, which dramatically limits the coverage area of wireless communications. To combat the high propagation attenuation, massive multiple-input multiple-output (MIMO) plays a key role in forming high-gain directional beams, which significantly extends the communication distances~\cite{Ahmed'18'survey}. 
\par In massive MIMO systems, the energy-efficient hybrid precoding architecture has been widely adopted in existing 5G communications to generate high-gain beams at bearable power consumption. Compared with conventional digital precoding schemes, hybrid precoding avoids employing a large number of power-consuming radio frequency (RF) by dividing a large-dimension digital precoder into a small-dimension digital precoder with RF chains and a large-dimension analog precoder implemented by phase shifters (PSs)~\cite{Rowell'15'mag,zhaoyaqiong'20'cl}. Due to the hardware constraint of analog PSs, additional constraints are imposed in analog precoding designs, enhancing the difficulty in designing hybrid precoding in wideband mmWave and THz communications~\cite{hanchong'jsac'20,Ayach'14'j}.
\subsection{Prior Works}\label{sec:intro prior}
\par Enormous works have investigated hybrid precoding methods in both narrowband and wideband systems. In narrowband mmWave communication systems, the initial work in~\cite{Ayach'14'j} proposed a compressed sensing (CS)-based precoding method, where an orthogonal matching pursuit (OMP) method was employed to determine the analog precoder. To further improve the spectrum efficiency performance, an iterative optimization method based on the Riemannian manifold is proposed in~\cite{Yu'16'j}. Nevertheless, the above-mentioned works assumed a fully-connected hybrid precoding network, which still introduces overwhelming hardware complexity. To solve this problem, the sub-connected architecture was proposed in~\cite{Gao'16'j}, where the successive interference cancellation (SIC) method is employed to derive closed-form digital and analog precoders. Moreover, a dynamic array-of-subarray architecture was further proposed, where the connections between analog precoders and RF chains can be dynamically adjusted to adapt to different channel conditions~\cite{hanchong'jsac'20}. 
% Different from the above methods assuming acquirement of full channel knowledge, reference~\cite{Alkhateeb'13'ita} further proposed a hybrid precoding design, where only the angle-of-arrival (AoA) information could be obtained at the base station. Moreover, the finite-resolution constraints of PSs were considered in the analog precoding designs in~\cite{Yu'15'ICASSP}. 
\par Although the above methods work well for narrowband communications, the hybrid precoding design experiences serious challenges in wideband mmWave communications~\cite{Molisch'17'mag}. Most existing works in wideband hybrid precoding designs rely on the PS network to adjust signal phase shifts at different antennas to form constructive interference at desired directions. In wideband communications, the required phase shifts to form constructive interference are frequency-dependent. Unfortunately, PSs could only generate the same phase shifts at different frequencies, which is called frequency-independent. This mismatch results in that PSs could only generate desired phase shifts at the central frequency~\cite{hanchong'22'jsac}, which is also termed the spatial wideband effect~\cite{Kim'21'twc}. As a result, at non-central frequencies the beams are squinted from the desired direction with linear arrays, resulting in severe beamforming loss. This effect was termed {\emph{beam squint}} effect in mmWave and became more significant in THz communications, which was identified as the {\emph{beam split}} effect. For illustration simplicity, in this paper we term these effects the beam split effect.
\par Existing works have attempted to overcome the beam split effect from the perspective of algorithm design and architecture design. In the works on algorithm designs, the hybrid precoding algorithm is elaborately designed to mitigate the beam split effect while the hardware architecture remains the same as in classical hybrid precoding. A frequency-selective hybrid precoding algorithm was proposed in~\cite{Alkhateeb'16'tcom}, where the analog precoder was exhaustively searched from a quantized codebook to maximize the achievable rate. To reduce the computational complexity, an efficient alternating optimization method based on channel statistics was proposed in~\cite{Kong'18'twc}. The method employing the channel covariance matrix over the wide bandwidth to efficiently design the analog precoder was proposed in~\cite{Heath'17'twc}. Moreover, the wide beams were also leveraged to combat the beam split effect~\cite{Liu'19'wcl}.
\par While in architecture designs, the true-time-delay devices which could generate frequency-dependent phase shifts were incorporated in the hybrid precoding architecture. Based on this architecture, a delay-phase precoding (DPP) precoding scheme was investigated in~\cite{Tan'22'twc}, where PSs and TTDs are combined to generate frequency-dependent beamforming. In this way, the beam split effect could be mostly eliminated. Moreover, a fixed TTD network was developed in~\cite{hanchong'22'jsac} to avoid the complicated adjustable TTDs.
\par Existing works in wideband hybrid precoding designs mainly focused on the deployment of uniform linear arrays (ULAs). Nevertheless, the effective array aperture of ULA dramatically reduces near the end-fire, resulting in reduced array gain~\cite{Navrati'16'}. Due to the feature of axial symmetry, the uniform circular array (UCA) has been regarded as a feasible solution to provide a wide scan-range and flat array gain with varying angles, which has been utilized to guarantee a better coverage area~\cite{Zhangjing'17'wc}.
\par In UCA communication systems, beamforming techniques have been investigated in narrowband systems~\cite{ma1974book,Feng'99'wpc,Kallnichev'01}. In~\cite{ma1974book} the steering vectors, the radiation pattern, and array directivity were thoroughly investigated. Leveraging the circular symmetry property of UCA, the estimation of the angle of arrival (AoA) has been studied in~\cite{Tewfik'90'c,Tewfik'92'j} in narrowband systems. To enable wideband communications in orthogonal frequency division multiplexing (OFDM) systems with UCA, the frequency invariant beamformer (FIB) was proposed in~\cite{Chan'07'tsp}, which aims to develop a digital beamformer to generate high array gain that is invariant over a wide range of frequencies. This method was further generalized into wideband UCA channel estimation~\cite{Zhang'17'j} and wideband indoor localization~\cite{Gentile'08'icc}.
% Recently, the FIB method was employed in near-field channel estimations with extremely large-scale UCA, where the FIB was utilized to capture the sparsity in the near-field region.
\par However, the aforementioned works on UCA only consider digital precoding schemes. In high-frequency bands, the fully-digital precoding architecture will bring overwhelming hardware costs and energy consumption. While the hybrid precoding scheme will introduce the spatial wideband effect, i.e. the mismatch of frequency-independent phase shifts of PSs and required frequency-dependent phase shifts to perform efficient beamforming~\cite{Kim'21'twc}. In this paper, we reveal that due to the variation of the array geometry, the spatial wideband effect will have a new manifestation different from the beam split effect, which brings severe beamforming loss in wideband UCA communications. To the best of our knowledge, such an effect has not been investigated in existing research, let alone the solutions to alleviate this effect.

\subsection{Our Contributions}\label{sec:intro contr}
In this paper, we first reveal the manifestation of the spatial wideband effect in UCA systems, the beam defocus effect, and characterize the beamforming loss resulting from the effect. Then, the DPP architecture incorporating TTD units is investigated to alleviate the beamforming loss. The contributions of the paper are listed as follows.
\begin{itemize}
	\item We first reveal that the spatial wideband effect is highly dependent on the array geometry and its manifestation in circular arrays is the beam defocus effect. Different from the beam split effect for ULA where beams are directed to different directions at different frequencies, high-gain beams could no longer form at non-central frequencies. This effect is fundamentally different from the beam split effect, which is termed the {\emph{beam defocus}} effect.
	\item The mechanism of the beam defocus effect is investigated, where the array pattern in both the frequency domain and the angular domain is characterized. By means of the series of Bessel functions, the beamforming gain in the frequency domain could be approximated with the zero-order Bessel function, indicating a severe beamforming loss in wideband hybrid precoding architectures with UCA. In addition, the beam pattern in the angular domain between beam defocus effect and beam split effect is compared, revealing their different manifestations.
	\item To mitigate the beam defocus effect, the DPP architecture is employed in UCA wideband communications. By leveraging the hybrid true-time-delay-phase-shifter (TTD-PS) network, the required frequency-dependent beams could be generated at a relatively low hardware complexity. The precoding design methods are also provided, which are demonstrated to be able to mitigate the beam defocus effect. Simulation results are finally provided to verify the effectiveness of the proposed DPP methods. 
\end{itemize}

\subsection{Organization and Notation}\label{sec:intro org}
\textit{Organization}: The remainder of the paper is organized as follows. Section~\ref{sec: sys} introduces the wideband UCA communication system model. Section~\ref{sec: beam dis} investigates the beam defocus effect in both frequency and angular domains. The DPP architecture is incorporated to alleviate the beam defocus effect in Section~\ref{sec: dpp arc}. The performance analysis is provided in Section~\ref{sec: sys per}. Simulation results are provided in Section~\ref{sec: sim}, and conclusions are drawn in Section~\ref{sec: con}.

\textit{Notations}: Lower-case and upper-case boldface letters represent vectors and matrices, respectively; $\mathbb{C}$ denotes the set of complex numbers; ${[\cdot]^{-1}}$, ${[\cdot]^{T}}$ and ${[\cdot]^{H}}$ denote the inverse, transpose and conjugate-transpose, respectively; $|\cdot|$ and $\| \cdot \|_2$ denote the norm and 2-norm of the argument, respectively; $[{\bf{A}}]_{i,j}$ denotes the $i^{\rm{th}}$ row and $j^{\rm{th}}$ column of matrix ${\bf{A}}$; ${\rm{blkdiag}}({\bf{A}})$ denotes a block diagonal matrix consisting of the columns of ${\bf{A}}$; $\mathcal{CN}(\mu,\sigma^2)$ denotes the complex Gaussian distribution with mean $\mu$ and variance $\sigma^2$.

\section{System Model}\label{sec: sys}
\par We consider a mmWave ELAA wideband communication system, where the classical hybrid precoding scheme is adopted. The base station (BS) is equipped with an $N$-element UCA while the users are assumed to be with an $N_{\rm{r}}$-element ULA. The BS employs $N_{\rm{RF}}$ RF chains, each of which connects to all antennas in a sub-connection manner~\cite{Gao'16'j}. To harvest the spatial multiplexing gain, multiple data streams are transmitted, satisfying $N_{\rm{s}} \leq N_{\rm{RF}} \leq N$. We assume $N_{\rm{s}} = N_{\rm{RF}} \ll N$ in this paper for illustration simplicity. The OFDM model with $M$ subcarriers is employed and the bandwidth is set to $B$. The central frequency is denoted by $f_c$ and the $m^{\rm{th}}$ subcarrier is denoted by $f_m=f_c+\frac{B(2m-1-M)}{2M}$. Thus, the received signal at the user side on the $m^{\rm{th}}$ subcarrier ${\bf{y}}_m \in \mathbb{C}^{N_r \times 1}$ could be expressed as
\begin{equation}
\label{eq: received signal}
\begin{aligned}
{\bf{y}}_m = \rho {\bf{H}}_m^H {\bf{F}}_{\rm{A}} {\bf{F}}_{{\rm{D}},m} {\bf{s}}_m + {\bf{n}}_m,
\end{aligned}
\end{equation}
where ${\bf{H}}_m \in \mathbb{C}^{N \times K}$ denotes the channel, ${\bf{F}}_{\rm{A}} \in \mathbb{C}^{N \times N_{\rm{RF}}}$ and ${\bf{F}}_{{\rm{D}},m} \in \mathbb{C}^{N_{\rm{RF}} \times N_s}$ denote the analog precoder and digital precoder, respectively. Note that since the digital precoder is conducted before digital-to-analog converters (DACs), it could be performed subcarrier by subcarrier and therefore is frequency-dependent. On the contrary, the analog precoder implemented by PSs introduces frequency-invariant phase shifts at different frequencies. Thus, the analog ${\bf{F}}_{\rm{A}}$ is frequency-independent. Due to the circuit restriction, the analog precoder is constrained with $|[{\bf{F}}_{\rm{A}}]_{i,j}| = \frac{1}{\sqrt{N}}$. The transmitted signal ${\bf{s}}_m \in \mathbb{C}^{N_s \times 1}$ and noise ${\bf{n}}_m \in \mathbb{C}^{N_r \times 1}$ follows ${\mathbb{E}}({\bf{s}}_m {\bf{s}}_m^H) = \frac{1}{N_s} {\bf{I}}_{N_s}$ and ${\bf{n}}_m \sim \mathcal{CN}(0,\sigma_n^2)$, respectively. 

\par Adopting the classical Saleh-Valenzuela channel model~\cite{Ayach'14'j}, The wireless channel could be written as
\begin{equation}
\label{eq: channel}
\begin{aligned}
{\bf{H}}_m = \sqrt{\frac{N}{L}} \sum_{l=1}^{L} g_l e^{-j2\pi \tau_l f_m} {\bf{a}}_m(\phi_{l}) {\bar{\bf{a}}}_m(\varphi_{l})^H,
\end{aligned}
\end{equation}
where ${\bf{a}}_m(\cdot)$ and ${\bar{\bf{a}}}_m(\cdot)$ denote the beam steering vectors at the BS side and user side, respectively. Notation $\phi_{l}$ and $\varphi_{l}$ denote the angle of departure and arrival of the $l^{\rm{th}}$ path, respectively. The expression of beam steering vector can be viewed as a frequency response vector for an impinging wave, which is highly dependent on the array structure. In this paper, UCA is employed at the BS shown in Fig.~\ref{img: UCA wideband}. Then, the beam steering vector corresponding to the physical direction $\phi$ can be written as~\cite{Kallnichev'01}
\begin{equation}
\label{eq: beam steering vector}
\begin{aligned}
{\bf{a}}_m(\phi) = \frac{1}{\sqrt{N}} \left[e^{j\eta_m \cos(\phi-\psi_0)}, \cdots, e^{j\eta_m \cos(\phi-\psi_{N-1})} \right]^T,
\end{aligned}
\end{equation}
where $\eta_m = \frac{2 \pi R f_m}{c}$ for $m=1,2,\cdots,M$ and $\psi_n = \frac{2\pi n}{N}$ for $n=0,1,\cdots,N-1$. For the user side, ULA is employed to fulfill the more strict requirements on the array deployment. Thus, the steering vector at the receiver could be expressed as
\begin{equation}
\label{eq: beam steering vector user}
\begin{aligned}
{\bar{\bf{a}}}_m(\varphi) = \frac{1}{\sqrt{N}} \left[1,e^{j\frac{2\pi df_m}{c}\sin\varphi}, \cdots, e^{j\frac{2\pi(N-1)df_m}{c}\sin\varphi} \right]^T.
\end{aligned}
\end{equation}
\begin{figure}[!t]
    \centering
    \setlength{\abovecaptionskip}{0.cm}
    \includegraphics[width=3in]{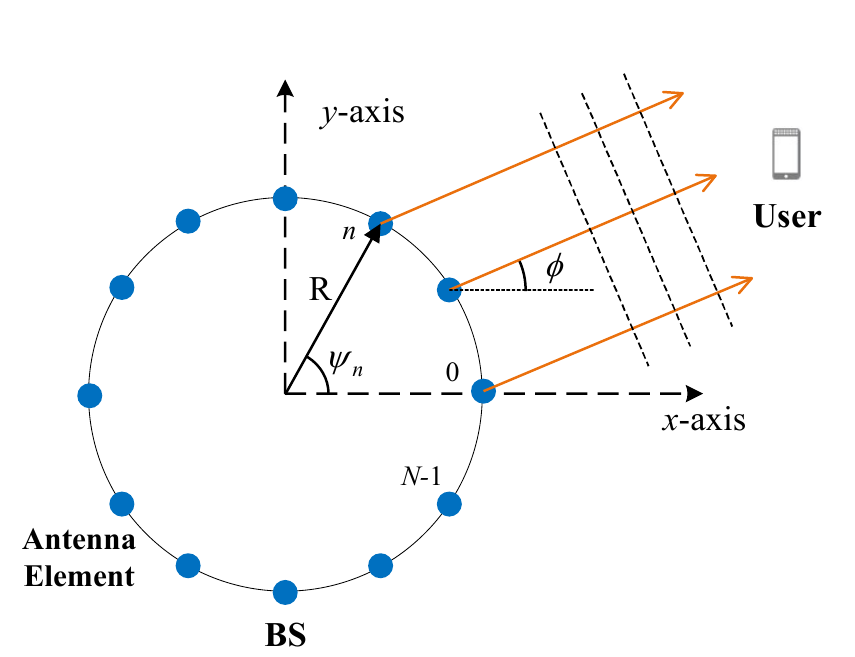}
    \caption{The geometry of UCA at BS.}
    \label{img: UCA wideband}
\end{figure}
\par It is worth noting that the beam steering vectors for constructing the channel are frequency-dependent, which means the phase differences between different antennas change with the frequency. This phenomenon is inconspicuous for narrowband systems, where the uniform phase shifts generated by PSs are feasible to perform ideal analog precoding~\cite{Tan'22'twc}. However, this phenomenon becomes obvious in wideband systems and will result in the beam defocus effect with the classical hybrid precoding architectures, which will be discussed in the following section.

\section{Beam Defocus Effect in UCA Systems}\label{sec: beam dis}
\par In this section, the mechanism of beam defocus for UCA systems will be first introduced. Then, the beamforming loss resulting from the beam defocus effect is characterized in both frequency and angular domains. 
\par In hybrid precoding architectures, the analog precoder and digital precoder are combined to generate directional and less interfered beams to harvest the spatial multiplexing gain and beamforming gain~\cite{Ahmed'18'survey}. Specifically, the analog precoding is aiming to form high-gain beams by performing constructive interference in desired directions. To achieve this goal, the phase shifters (PSs) are elaborately designed to compensate for the phase differences between different antennas to form equal-phase planes. However, it has been revealed that PS could only generate frequency-independent phases. On the contrary, to form directional beams focused on the same direction, the required phase shifts are frequency-dependent, which has been reflected by the term $\eta_m$ in the beam steering vector in equation~\eqref{eq: beam steering vector}. This mismatch of the frequency-independent phase shifts generated by PSs and the required frequency-dependent phase shifts is often termed the spatial wideband effect.
% This phenomenon is illustrated in Fig. X.
\par Recently, it has been revealed that the spatial wideband effect has given rise to the beam split effect in ULA systems~\cite{Tan'22'twc}, which can be viewed as a manifestation of the spatial wideband effect corresponding to the linear array geometry. Nevertheless, the beam pattern analysis of the beam split effect in~\cite{Tan'22'twc} is only restricted to ULA. As the array geometry varies from linear to circular, the spatial wideband effect is reflected in another manifestation, which shall be discussed as follows. 
\begin{lemma}
\label{lemma1}
Employing frequency-independent phase shifts which are designed according to the central frequency $f_c$, the achieved beamforming gain at the subcarrier $f_m$ at the desired direction $\phi$ could be expressed as  
\begin{equation}
\label{eq: gain wideband}
\begin{aligned}
G_m({\bf{a}}_c(\phi),\phi) &= \left|{\bf{a}}_m^H(\phi){\bf{a}}_c(\phi) \right|\\
 &= \left| \frac{1}{N}\sum_{n=0}^{N-1} e^{-j(\eta_m-\eta_c) \cos(\phi-\psi_n)} \right|\\
 &\approx \left| J_0\left(\frac{2\pi R (f_m-f_c)}{c}\right) \right| = \left| J_0(\eta_m-\eta_c) \right|,
\end{aligned}
\end{equation}
where $J_0(\cdot)$ denotes the zero-order Bessel function of the first kind, $\eta_m = \frac{2\pi R f_m}{c}$ and $\eta_c = \frac{2\pi R f_c}{c}$.
\end{lemma}
\begin{proof}
The proof process starts from the generating function of Bessel functions~\cite{bowman2012introduction}, which is written as
\begin{equation}
\label{eq: lemma1 eq1}
\begin{aligned}
e^{j\beta \cos\gamma} = \sum_{s=-\infty}^{\infty} j^s J_s(\beta) e^{js\gamma},
\end{aligned}
\end{equation}
where $J_s(\cdot)$ denotes the $s$-order Bessel function of the first kind. By substituting~\eqref{eq: lemma1 eq1} into~\eqref{eq: gain wideband}, we can obtain
\begin{equation}
\label{eq: lemma1 eq2}
\begin{aligned}
G_m({\bf{a}}_c(\phi),\phi) \mathop{=}\limits^{(a)} \frac{1}{N} \left| \sum_{s=-\infty}^{\infty} j^s J_s(\eta_c-\eta_m) e^{js\phi} \sum_{n=0}^{N-1} e^{-js\psi_n} \right|,
\end{aligned}
\end{equation}
where equation (a) is derived by exchanging the order of summation. Then, the summation over $n$ could be expressed as the piecewise function as
\begin{equation}
\label{eq: lemma1 eq3}
\sum_{n=0}^{N-1} e^{-js\psi_n}=\left\{
\begin{aligned}
N, \quad &s = N \cdot t, t \in \mathbb{Z} \\
0, \quad &s \neq N \cdot t, t \in \mathbb{Z}.\\
\end{aligned}
\right.
\end{equation}
The piecewise function reveals the periodicity of the summation, which equals zero except on integral multiples of $N$. Then, we assume that $N$ is large enough, which has been a common assumption in the analysis for UCA~\cite{Zhang'17'tawp}. Following the inequation that
\begin{equation}
\label{eq: lemma1 eq4}
\begin{aligned}
\left| J_{|s|}(x) \right| \leq \left( \frac{x e}{2|s|} \right) ^{|s|},
\end{aligned}
\end{equation}
the value of $|J_{|s|}(x)|$ becomes negligible for large $s$. Therefore, $|J_{|s|}(x)| \approx 0$ could be assumed for $s = N \cdot t$ with $t \neq 0$ and large $N$. Finally, the summation could obtain an accurate approximation with the term $J_{0}(\eta_c-\eta_m)$. Finally, due to the symmetry of $J_0(x)$, the beamforming gain could be approximated as
\begin{equation}
\label{eq: lemma1 eq5}
\begin{aligned}
G_m({\bf{a}}_c(\phi),\phi) \approx \left| J_0(\eta_c-\eta_m) \right| = \left| J_0(\eta_m-\eta_c) \right|,
\end{aligned}
\end{equation}
which completes the proof.
\end{proof}

\begin{remark}
{\bf{Lemma}~\ref{lemma1}.} has characterizes the beamforming loss resulting from the frequency-independent phase shifts in the frequency domain. According to the property of $J_0(x)$, the beamforming gain employing PSs could only achieve its maximum when $f_m = f_c$, i.e. at the central frequency. The perfect constructive interference could not form at any subcarrier other than the central frequency, introducing unwilling beamforming loss across the wide bandwidth. In addition, the beamforming gain is not dependent on the angle $\phi$, which is resulted from the rotational symmetry of UCA. This is also one of the distinguishing features of UCAs.
\end{remark}
\par The beamforming gain at different frequencies is plotted in Fig.~\ref{img: beam frequency}. The central frequency is set to $30$ GHz with the bandwidth $B=4$ GHz. The BS is equipped with a $256$-element half-wavelength spaced UCA. It is shown that the optimal beamforming could only be achieved at the central frequency, which is consistent with the analysis in {\bf{Lemma}~1}. According to the overall downtrend of $|J_0(x)|$, a larger frequency deviation from the central frequency will result in a more severe beamforming loss. In addition, the destructive interference could even be formed at certain frequencies, resulting in zero gain.
\begin{figure}[!t]
    \centering
    \setlength{\abovecaptionskip}{0.cm}
    \includegraphics[width=3in]{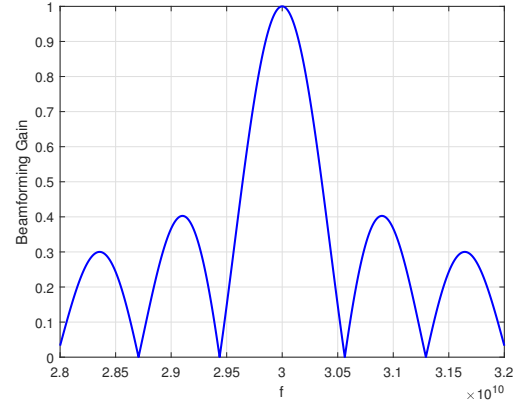}
    \caption{The beamforming gain achieved with frequency-independent phase shifts at different frequencies.}
    \label{img: beam frequency}
\end{figure}
\par The above analysis has indicated that wideband communications with UCA will suffer from severe loss originating from the constraints of frequency-independent PSs in hybrid precoding. It may seem that this phenomenon is very much like the beam split effect for ULA. Next, we shall illustrate the distinction of the beam defocus effect, which significantly distinguishes itself from the beam split effect with the following lemma.

\begin{lemma}
\label{lemma2}
If the frequency-independent beam steering vector ${\bf{a}}_c(\phi_0)$ is employed, the achieved beamforming gain at frequency $f_m$ at any direction $\phi$ could be expressed as  
\begin{equation}
\label{eq: gain wideband angular}
\begin{aligned}
G_m({\bf{a}}_c(\phi_0),\phi) &= \left|{\bf{a}}_m^H(\phi){\bf{a}}_c(\phi_0) \right|\\
 &\approx \left| J_0(\xi) \right|,
\end{aligned}
\end{equation}
where the parameter $\xi$ is defined as
\begin{equation}
\label{eq: xi}
\begin{aligned}
\xi = \sqrt{\eta_m^2+\eta_c^2-2\eta_m\eta_c\cos(\phi-\phi_0)}.
\end{aligned}
\end{equation}
\end{lemma}
\begin{proof}
The proof is provided in {\bf Appendix~\ref{app: lemma2}}.
\end{proof}
\par This lemma characterizes the beam pattern of the beam defocus effect in the angular domain. According to the property of $J_0(\cdot)$, the beamforming gain could only achieve the maximum when $\xi = 0$, i.e. $f_m = f_c$ and $\phi = \phi_0$ are simultaneously satisfied. Note that if $\phi = \phi_0$ is satisfied, the result in {\bf{Lemma}~\ref{lemma2}} will degrade into {\bf{Lemma}~\ref{lemma1}}. Therefore, {\bf{Lemma}~\ref{lemma2}} depicts a more general case where $\phi$ can be arbitrarily selected. For fixed $f_m \neq f_c$, beamforming gain at $f_m$ could not reach 1 in any direction. As a consequence, high-gain beams with optimal beamforming gain could not be formed except at the central frequency. This phenomenon is fundamentally different from the beam split effect for ULA systems, where the high-gain beams with optimal beamforming gain split into separated physical directions.

\par To clearly illustrate the differences between the beam defocus and beam split effect, a comparison of the beam pattern in the angular domain is shown in Fig.~\ref{img: beam compare}, where different colored lines represent the beam pattern at different frequencies. Specifically, the high-gain beams slightly squint from the desired direction but retain the same beam pattern in the beam split effect, as shown in Fig.~\ref{img: beam split ill}. On the contrary, the beam pattern is severely distorted at non-central frequencies in the beam defocus effect, as shown in Fig.~\ref{img: beam defocus ill}. The high-gain beams no longer exist in any direction, which is consistent with the result in {\bf{Lemma}~\ref{lemma2}}. Since this effect is similar to the defocus phenomenon in photography, we name it the beam defocus effect. In addition, the colored solid lines in Fig.~\ref{img: beam defocus ill} denote the accurately calculated beamforming gain across the bandwidth while the black dashed lines denote the estimated beamforming gain with {\bf{Lemma}~2}. The consistency of the solid and dashed lines indicates that the estimation in~\eqref{eq: gain wideband angular} could achieve a high degree of accuracy.
 \begin{figure}[!t]
    \centering
    \subfigure[Beam Split Effect]{
    	\includegraphics[width=2.5in]{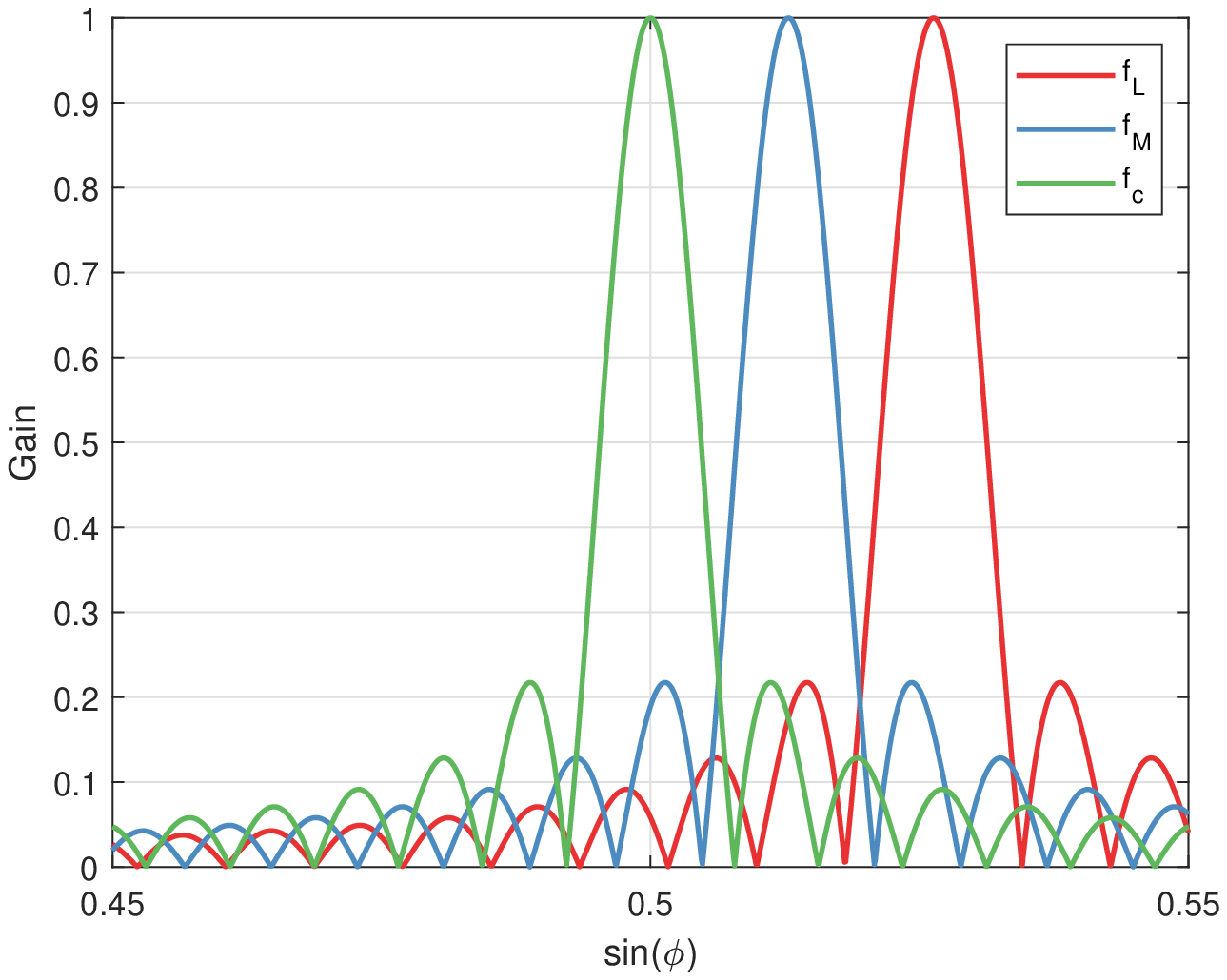}
    	\label{img: beam split ill}}
    \subfigure[Beam Defocus Effect]{
    	\includegraphics[width=2.5in]{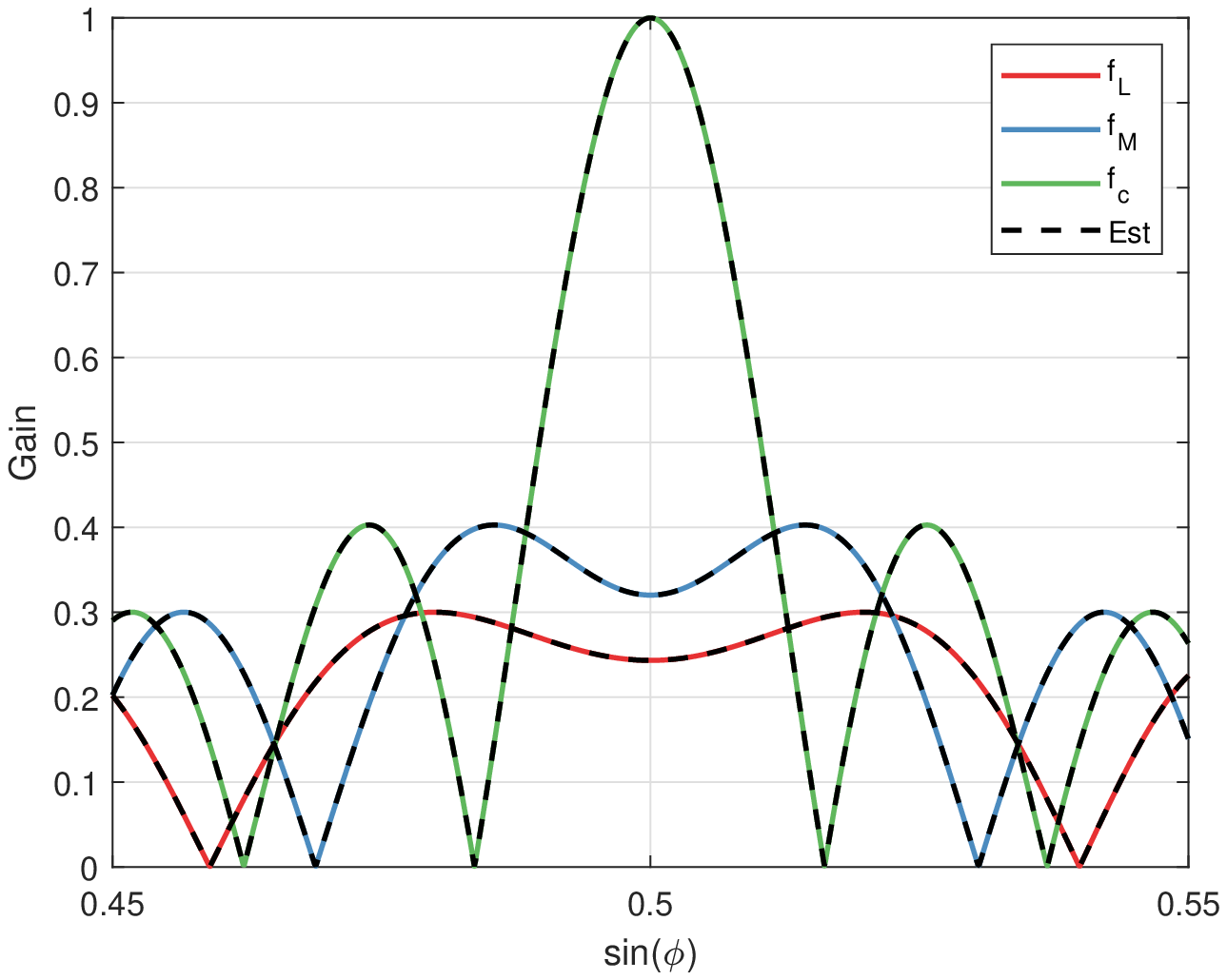}
    	\label{img: beam defocus ill}}
    \caption{Comparison over the beam split effect and beam defocus effect in the angular domain. Both systems are designed to generate beams towards $\sin(\phi) = 0.5$ at the central frequency. $f_L=28.5$ GHz and $f_c=30$ GHz denote the lowest frequency and central frequency while $f_M$ denotes the middle point as $f_M=\frac{f_L+f_c}{2}$.}
    \label{img: beam compare}
    \vspace{-3mm}
\end{figure}
\begin{remark}
A question naturally raises, given that no high-gain beams are generated at non-central frequencies for UCA, where has the transmitted signal power gone in the beam defocus? In fact, the optimal beamforming gain originates from perfect constructive interference, where the phase differences from different antennas are compensated through analog precoding. Due to the mismatch between generated frequency-independent phase shifts and required frequency-dependent phase shifts, no perfect constructive interference will happen in any direction. Therefore, the pencil-like beams no longer exist like in the beam split effect. The transmitted signal power does not change, but spreads into a wide range. Through this phenomenon, we would like to highlight in this paper that the beam pattern is highly dependent on the array geometry, and the beam split effect is not the only manifestation of the spatial wideband effect.
\end{remark}

\section{DPP Architecture to Alleviate the Beam Defocus Effect}\label{sec: dpp arc}% misfocus, defocus
\par To retrieve the beamforming loss resulting from the beam defocus effect, the delay-phase precoding (DPP) architecture is introduced in this section. The precoding algorithm will be provided and the beamforming gain when employing DPP will be investigated.

\subsection{DPP Architecture}\label{sec: dpp}
As discussed in previous sections, severe beamforming loss will be introduced in classical wideband hybrid precoding architectures due to the beam defocus effect, which remarkably worsens the received signal quality across the whole bandwidth. As a result, the wideband UCA communication system will suffer from serious performance loss in PS-based hybrid precoding architectures. 
 \begin{figure}[!t]
    \centering
    \subfigure[Classical Hybrid Precoding Architecture]{
    	\includegraphics[width=3.2in]{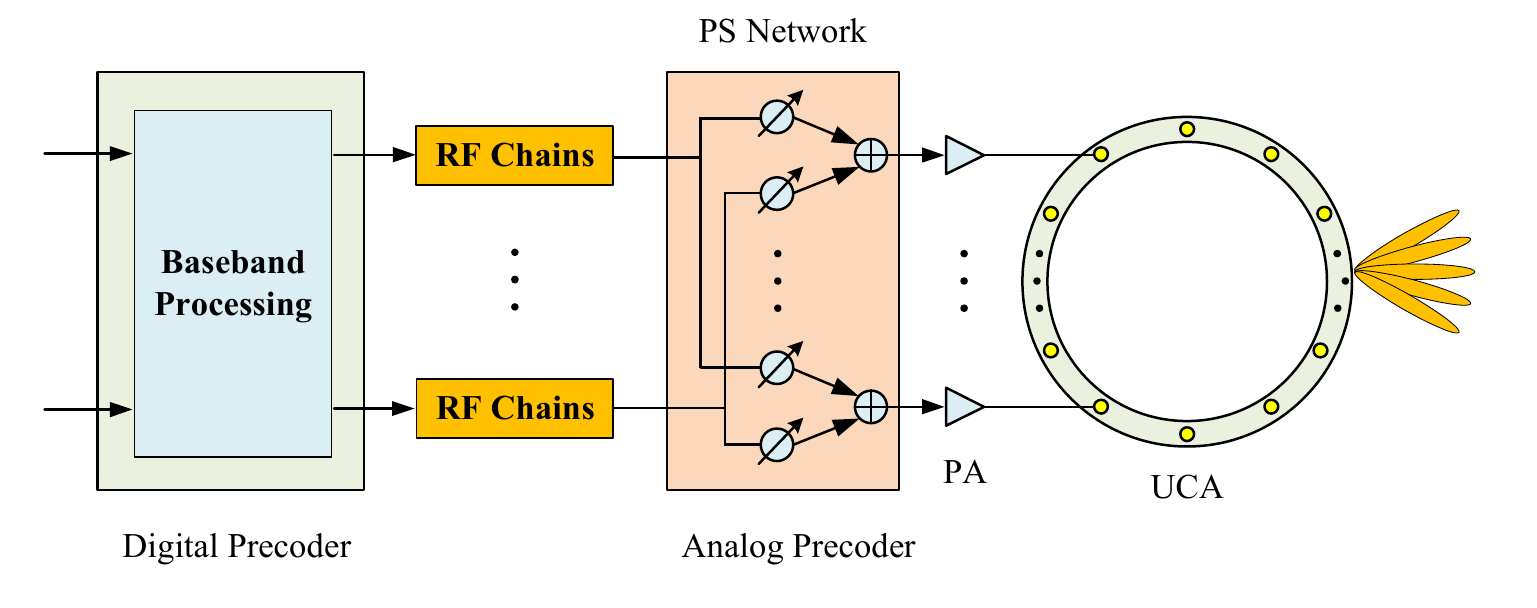}
    	\label{img: hybrid pre}}
    \subfigure[Delay-Phase Precoding Architecture]{
    	\includegraphics[width=3.5in]{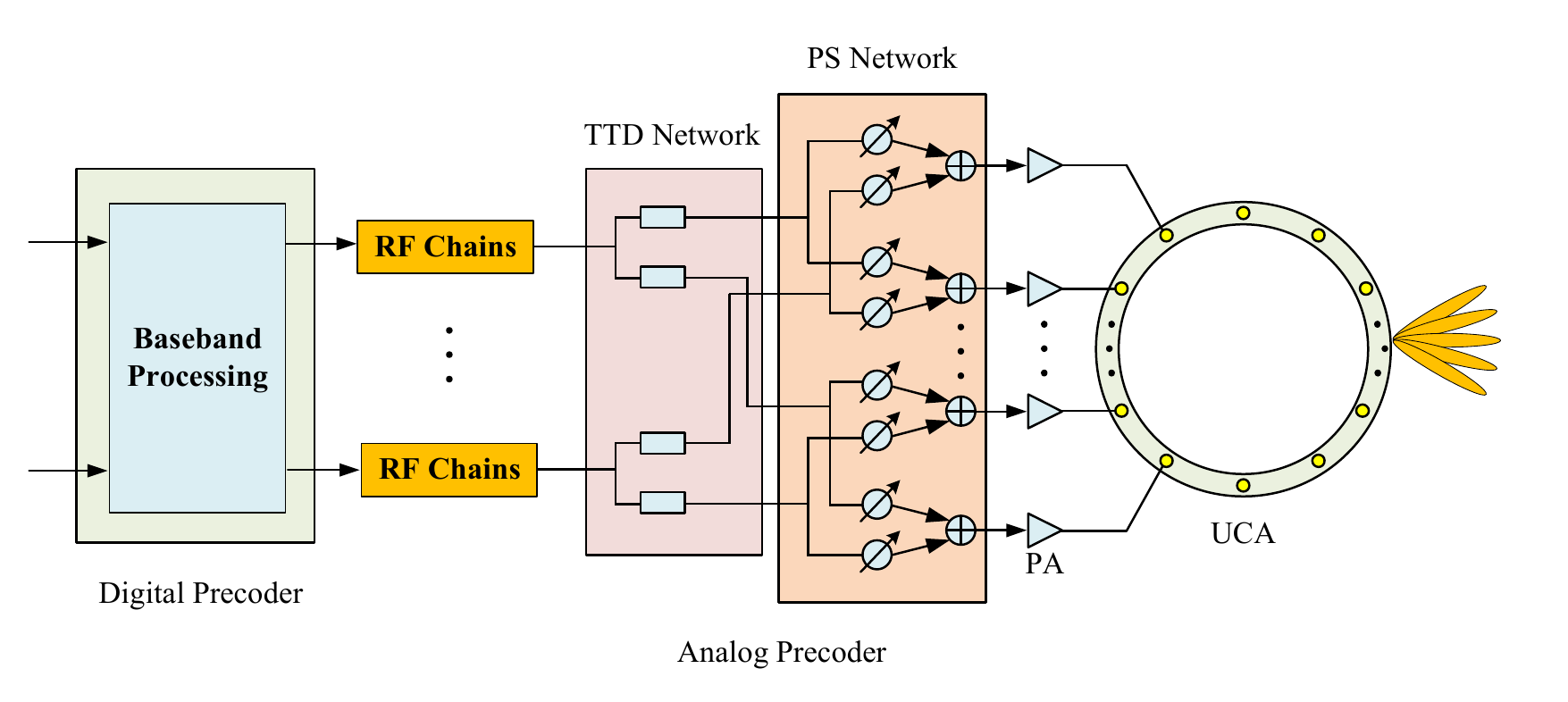}
    	\label{img: DPP arch}}
    \caption{Comparison of the classical hybrid precoding architecture and employed DPP architecture.}
    \label{img: architecture compare}
    \vspace{-3mm}
\end{figure}
\par To break the hardware constraint of PS that could only generate frequency-independent phase shifts, TTD devices have been incorporated before the PS network. In this architecture, the delay of TTD devices and the phase shifts of PSs are combined to produce frequency-dependent phase shifts to cope with the beam split effect in ULA systems~\cite{Tan'22'twc}, which is also termed the DPP architecture. 
\par Noting that despite different manifestations, beam split effect and beam defocus effect share the same root cause, that is the mismatch of the generated frequency-independent phase shifts and required frequency-dependent phase shifts. Inspired by this discovery, we propose to generalize the DPP architecture originally employed in ULA systems into UCA systems to alleviate the beam defocus effect. Different from the classical hybrid precoding architecture as in Fig.~\ref{img: hybrid pre}, the DPP architecture additionally introduces a TTD network to generate frequency-dependent phase shifts, as shown in Fig.~\ref{img: DPP arch}. PSs and TTD units are combined to perform analog precoding. In this paper, we assume that each RF chain connects to $K$ TTD units and each TTD unit connects to $P=\frac{N}{K}$ antennas in a sub-connected manner.
\par Based on the DPP architecture, the TTD-PS analog precoder corresponding to the $l^{\rm{th}}$ component of the channel could be expressed as
\begin{equation}
\label{eq: dpp bf vec}
\begin{aligned}
{\bf{b}}_{m,l} = {\rm{blkdiag}}({\bf{F}}_l){\bf{p}}_{l,m},
\end{aligned}
\end{equation}
where ${\bf{F}}_l = [{\bf{f}}_{l,1},{\bf{f}}_{l,2},\cdots,{\bf{f}}_{l,K}]$ and ${\bf{p}}_{l,m}$ represent the phase shifts generated by the $K$ TTD units. Note that each column of ${\bf{F}}_l$ corresponds to one TTD unit, the element of which has the constant modulus constraint $|[{\bf{f}}_{l,i}]_j|=\frac{1}{\sqrt{N}}$. Roughly speaking, the phase shifts generated by TTD units are linear to the frequency\footnote{It is worth noting that a pure delay of TTD results in phase shifts in proportion to frequency~\cite{Tan'22'twc}. Nevertheless, the frequency-independent term could be realized through PSs, which will be introduced in the next subsection. In this subsection, we relax the constraint of TTD units from in proportion to frequency to linear to frequency for illustration simplicity.}.

\par Then, the analog precoding design could be decomposed into two components, i.e. frequency-independent term ${\bf{f}}_{l,i}$ and frequency-dependent term ${\bf{p}}_{l,m}$. As a consequence, the objective of analog precoding design is to maximize the beamforming gain, which is expressed by
\begin{equation}
\label{eq: dpp bf vec2}
\begin{aligned}
\mathop{\max}\limits_{{\bf{b}}_{m,l}}~ G_m({\bf{b}}_{m,l} ,\phi_l) = \mathop{\max}\limits_{{\bf{b}}_{m,l}}~ \left|{\bf{a}}_m^H(\phi_l){\bf{b}}_{m,l}\right|.
\end{aligned}
\end{equation}
\par We note that PSs could work well in narrowband systems but fail to meet the requirement of frequency-dependent phase shifts in wideband systems. A natural method is first designing PSs to align the beam towards the desired direction at the central frequency, like in narrowband systems~\cite{Tan'22'twc}. Then, the TTDs are designed to compensate for the frequency-dependent residual components, which is critical to mitigating the beam defocus effect in wideband systems. Following this method, ${\bf{F}}_l$ could be designed according to
\begin{equation}
\label{eq: phase shifter item}
\begin{aligned}
{\bf{F}}_l = \left[{\bf{a}}_c(\phi_l)_1,{\bf{a}}_c(\phi_l)_2,\cdots, {\bf{a}}_c(\phi_l)_K \right],
\end{aligned}
\end{equation}
where ${\bf{a}}_c(\phi_l)_i \in {\mathbb{C}}^{P \times 1}$ denotes the $i^{\rm{th}}$ subvector of ${\bf{a}}_c(\phi_l)$. Then, the delay design of TTD units could be obtained through the following lemma.
\begin{lemma}
\label{lemma3}
The optimal designed ${\bf{p}}_{l,m}$ compensating for the frequency-dependent residuals of the $l^{\rm{th}}$ channel component could be expressed as
\begin{equation}
\label{eq: lemma3 eq1}
\begin{aligned}
[{\bf{p}}_{l,m}]_k = \exp\left\{j\frac{2\pi R}{c}(f_m-f_c)\cos\left(\phi_l-\frac{\pi(2k-1)}{K}\right)\right\}, 
\end{aligned}
\end{equation}
where $[{\bf{p}}_{l,m}]_k$ denotes the $k^{\rm{th}}$ element of ${\bf{p}}_{l,m}$. The corresponding beamforming gain at frequency $f_m$ is written as
\begin{equation}
\label{eq: lemma3 eq2}
\begin{aligned}
G_m({\bf{b}}_{m,l} ,\phi_l) \approx \frac{1}{P} \sum_{i=0}^{P-1}J_0(R_i), 
\end{aligned}
\end{equation}
where $R_i = \frac{2\sqrt{2}\pi R}{c}(f_m-f_c)\sqrt{1-\cos\left(\frac{(2i+1)\pi}{N}-\frac{\pi}{K}\right)}$ for $i=0,1,\cdots,P-1$.
\end{lemma}

\begin{proof}
The proof is provided in {\bf Appendix~\ref{app: lemma3}}.
\end{proof}
\par Note that a large $K$ is assumed in the proof of {\bf{Lemma}~3} to acquire a relatively accurate approximation. Through this lemma, we can see that the beamforming gain after introducing the TTD network is highly dependent on the choice of $K$. Since it is hard to directly extract the relationship between the beamforming gain and $K$, we seek a more succinct expression of beamforming gain with the following corollary.
\begin{corollary}
\label{coro1}
With the assumption of a large $K$, the beamforming gain with DPP architectures obtained in~\eqref{eq: lemma3 eq2} could be further simplified as
\begin{equation}
\label{eq: coro1 eq1}
\begin{aligned}
G_m({\bf{b}}_{m,l} ,\phi_l) \approx \left|_1F_2\left(\frac{1}{2};1,\frac{3}{2};-\frac{a^2}{4}\right) \right|,
\end{aligned}
\end{equation}
where the notation $ _1F_2$ represents the generalized hypergeometric function~\cite{mathai2006generalized} and $a = \frac{2\pi^2 R}{cK}(f_m-f_c)$.
\end{corollary}
\begin{proof}
The proof is provided in {\bf Appendix~\ref{app: coro1}}.
\end{proof}
Although the special function $_1F_2(\cdot)$ in~\eqref{eq: coro1 eq1} seems complicated, it is still a one-variable function over $a$, where an example of $_1F_2(\frac{1}{2};1,\frac{3}{2};-\frac{x^2}{4})$ against $x$ is plotted as the red line in Fig.~\ref{img: hyper function}. According to the overall downtrend of $_1F_2(\frac{1}{2};1,\frac{3}{2};-\frac{x^2}{4})$, we can sketchily conclude that a larger number of TTDs is required if we aim to further mitigate the beamforming loss. Then, a critical problem is to estimate the number of TTD units required if we want to ensure a beamforming loss less than a predetermined threshold across the whole bandwidth, which will be discussed in the following corollary.
\begin{figure}[!t]
    \centering
    \setlength{\abovecaptionskip}{0.cm}
    \includegraphics[width=3in]{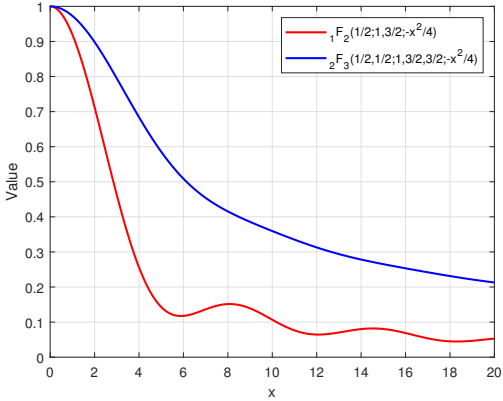}
    \caption{Illustration of the generalized hypergeometric functions $_1F_2(\frac{1}{2};1,\frac{3}{2};-\frac{x^2}{4})$ and $_2F_3(\frac{1}{2},\frac{1}{2};1,\frac{3}{2},\frac{3}{2};-\frac{x^2}{4})$ against $x$.}
    \label{img: hyper function}
\end{figure}

\begin{corollary}
\label{coro2}
Aiming to achieve a beamforming loss less than a threshold $\Delta$ in the whole bandwidth, the required number of TTD units has to satisfy
\begin{equation}
\label{eq: coro2 eq1}
\begin{aligned}
K \geq \frac{\pi^2 RB}{cF^{-1}(1-\Delta)},
\end{aligned}
\end{equation}
where the function $F^{-1}(1-\Delta)$ is defined by $F^{-1}(1-\Delta) \triangleq {\underset {x} { \operatorname {arg\,min} } \, \left\{ _1F_2(\frac{1}{2};1,\frac{3}{2};-\frac{x^2}{4}) = 1-\Delta \right\}}$.
\end{corollary}
\begin{proof}
Note that $|f_c-f_m|\leq B/2$ always holds for any subcarrier $f_m$. By substituting the definitions of $a$ and $F^{-1}(\cdot)$ into~\eqref{eq: coro1 eq1}, the results in~\eqref{eq: coro2 eq1} can be easily obtained. 
\end{proof}

\subsection{DPP Algorithm for UCA Systems}\label{sec: dpp precoding}
\par In the previous subsection, the beamforming gain with TTD-PS-based analog precoding has been analyzed. In this subsection, the complete precoding algorithm consisting of the digital precoder and analog precoder implemented by TTD units will be provided.
\par Based on the TTD-PS architecture, the received signal can be rewritten as
\begin{equation}
\label{eq: received signal TTD}
\begin{aligned}
{\bf{y}}_m = \sqrt{\rho} {\bf{H}}_m^H {\bf{F}}_{\rm{A}}^{\rm{PS}} {\bf{F}}_{{\rm{A}},m}^{\rm{TTD}} {\bf{F}}_{{\rm{D}},m} {\bf{s}}_m + {\bf{n}}_m,
\end{aligned}
\end{equation}
where ${\bf{F}}_{\rm{A}}^{\rm{PS}}$ and ${\bf{F}}_{{\rm{A}},m}^{\rm{TTD}}$ denote the analog precoding corresponding to the PS network and TTD network, respectively. According to the functional principle of TTD units, the phase shift generated by TTD has to be strictly proportional to frequency, written as $-2\pi f_m t$ where $t$ denotes the delay. Noticing that the optimal phase shifts required for the $k^{\rm{th}}$ TTD in equation~\eqref{eq: lemma3 eq1} have two separated components, written as
\begin{equation}
\label{eq: phase division}
\begin{aligned}
\angle [{\bf{p}}_{l,m}]_k^* &= -\frac{2\pi R}{c}\cos\left(\phi_l-\frac{\pi(2k+1)}{K}\right)f_c\\
&+\frac{2\pi R}{c}\cos\left(\phi_l-\frac{\pi(2k+1)}{K}\right)f_m,
\end{aligned}
\end{equation}
where the former is frequency-independent while the latter is exactly proportional to frequency $f_m$. Therefore, we can divide the required phase shifts into the frequency-independent component and frequency-dependent component, which can be realized by PS and TTD, respectively. Thus, the required time delay for the $k^{\rm{th}}$ TDD corresponding to the $l^{\rm{th}}$ channel component has to satisfy
\begin{equation}
\label{eq: TTD design}
\begin{aligned}
-2\pi f_m t_{l,k} = \frac{2\pi R}{c}\cos\left(\phi_l-\frac{\pi(2k+1)}{K}\right)f_m.
\end{aligned}
\end{equation}
\par Then, we can obtain $t_{l,k} = -\frac{R}{c}\cos\left(\phi_l-\frac{\pi(2k+1)}{K}\right)$. In addition, since the TTD has an additional constraint $t_l>0$, a global time delay needs to be added to all TTD units without influencing the beamforming performance. The modified delay could be written as
\begin{equation}
\label{eq: TTD design modify}
\begin{aligned}
{\widetilde{t}}_{l,k} = \frac{R}{c}\left(1-\cos\left(\phi_l-\frac{\pi(2k+1)}{K}\right)\right).
\end{aligned}
\end{equation}
Finally, the analog precoder corresponding to the TTD network is written as
\begin{equation}
\label{eq: TTD beamformer}
\begin{aligned}
{\widetilde{\bf{p}}}_{l,m} = [e^{-2\pi f_m {\widetilde{t}}_{l,1}},e^{-2\pi f_m {\widetilde{t}}_{l,2}},\cdots,e^{-2\pi f_m {\widetilde{t}}_{l,K}}]^T.
\end{aligned}
\end{equation}
% Finally, the analog precoder corresponding to the TTD is written as
% \begin{equation}
% \label{eq: ana TTD}
% \begin{aligned}
% {\bf{F}}_{{\rm{A}},m}^{\rm{TTD}} = {\rm{blkdiag}} \left({\widetilde{\bf{p}}}_{1,m},{\widetilde{\bf{p}}}_{2,m},\cdots,{\widetilde{\bf{p}}}_{N_{\rm{RF}},m}\right).
% \end{aligned}
% \end{equation}
\par So far, the frequency-dependent component of phase shifts has been obtained. After combining the frequency-independent component in~\eqref{eq: phase division} into the original analog precoder ${\bf{F}}_l$ in~\eqref{eq: phase shifter item}, the modified analog precoder corresponding to the PS network ${\widetilde{\bf{F}}_l}$ can be written as
\begin{equation}
\label{eq: ana ps 2}
\begin{aligned}
{\widetilde{\bf{f}}}_{l,k} = {\bf{a}}_c(\phi_l)_k e^{-\frac{2\pi R}{c}\cos\left(\phi_l-\frac{\pi(2k+1)}{K}\right)f_c},
\end{aligned}
\end{equation}
where ${\widetilde{\bf{f}}}_{l,k}$ denotes the $k^{\rm{th}}$ column of ${\widetilde{\bf{F}}_l}$. Finally, the analog precoder ${\bf{F}}_{{\rm{A}}}^{\rm{PS}}$ can be constructed by concatenating all analog precoders ${\widetilde{\bf{F}}_l}$.
% \begin{equation}
% \label{eq: ana ps}
% \begin{aligned}
% {\bf{F}}_{{\rm{A}}}^{\rm{PS}} = \left[{\rm{blkdiag}}({\widetilde{\bf{F}}_1}), \cdots, {\rm{blkdiag}}({\widetilde{\bf{F}}_{N_{\rm{RF}}}}) \right].
% \end{aligned}
% \end{equation}
\par The employed DPP algorithm is summarized in {\bf{Algorithm}~\ref{alg:1}}. Note that each RF chain is designed to cope with one channel path. We first rearrange the channel components in descending order, dealing with the most significant channel components with limited RF chains, as shown in line 1. Then, according to the separation method of the required phase shift in~\eqref{eq: phase division}, we can add phase shifts to classical beam steering vector to obtain each column of the PS-based analog precoder, as shown in lines 3, 5, and 8. By extracting the frequency-dependent components, we can obtain the analog precoder associated with TTD network in each RF chain, as shown in lines 6 and 9. Then, by concatenating the derived precoders for each RF chain, the analog precoder can be obtained in line 11-12. Finally, we can follow the classical water-filling procedure, using the equivalent channel and singular value decomposition (SVD) to obtain digital precoder, as shown in lines 13-14~\cite{Heath'15'j}. 

\begin{remark}
It is worth noting that, the acquirement of every channel component is a precondition of the proposed DPP algorithm, which is often adopted as a basic condition in existing works~\cite{Tan'22'twc,hanchong'22'jsac,chenzhi'23'open}. Generally, the channel state information could be effectively obtained through compressive sensing methods~\cite{Mojianhua'18'tsp} or deep learning methods~\cite{hanchong'21'tcom}.
\end{remark}
\begin{algorithm*}[!t] 
	\caption{DPP algorithm for UCA.} 
	\label{alg:1} 
	\begin{algorithmic}[1] %这个1 表示每一行都显示数字
		\REQUIRE ~%算法的输入参数：Input
		Channel ${\bf{H}}_m$, number of antennas $N$, RF chains $N_{\rm{RF}}$, TTD units $K$, UCA radius $R$, frequency $f_m$ and angles $\phi_l$
		\ENSURE ~%算法的输出：Output
        Analog precoder ${\bf{F}}_{\rm{A}}^{\rm{PS}}$ and ${\bf{F}}_{{\rm{A}},m}^{\rm{TTD}}$, digital precoder ${\bf{F}}_{{\rm{D}},m}$
		\STATE Rearrange the order of the channel components $|g_1| \geq |g_2| \geq \cdots \geq |g_{\rm{RF}}|$ and obtain corresponding $\{\phi_1,\cdots,\phi_{\rm{RF}}\}$;
		\FOR{$l = 1,2,\cdots,N_{\rm{RF}}$}
		\STATE Construct the classical beam steering vector ${\bf{a}}_c(\phi_l) = \frac{1}{\sqrt{N}}\left[e^{j\eta_c \cos(\phi_l-\psi_0)}, \cdots, e^{j\eta_c \cos(\phi_l-\psi_{N-1})}\right]^T$;
		\FOR{$k = 1,2,\cdots,K$}
		\STATE Append additional frequency-independent phase shifts ${\widetilde{\bf{f}}}_{l,k} = {\bf{a}}_c(\phi_l)_k e^{-j\frac{2\pi R}{c}\cos\left(\phi_l-\frac{\pi(2k+1)}{K}\right)f_c}$;
		\STATE Determine the required delay for $k^{\rm{th}}$ TTD ${\widetilde{t}}_{l,k} = \frac{R}{c}\left(1-\cos\left(\phi_l-\frac{\pi(2k+1)}{K}\right)\right)$;
		\ENDFOR
		\STATE Construct the PS-based analog precoder for the $l^{\rm{th}}$ RF chain ${\widetilde{\bf{F}}}_{l} = [{\widetilde{\bf{f}}}_{l,1},\cdots,{\widetilde{\bf{f}}}_{l,K}]$;
		\STATE Construct the TTD-based analog precoder for the $l^{\rm{th}}$ RF chain ${\widetilde{\bf{p}}}_{l,m} = [e^{-j2\pi f_m {\widetilde{t}}_{l,1}},e^{-j2\pi f_m {\widetilde{t}}_{l,2}},\cdots,e^{-j2\pi f_m {\widetilde{t}}_{l,K}}]^T$;
		\ENDFOR
		\STATE Concatenate PS-based analog precoders ${\bf{F}}_{\rm{A}}^{\rm{PS}} = \left[{\rm{blkdiag}}({\widetilde{\bf{F}}}_{1}),\cdots,{\rm{blkdiag}}({\widetilde{\bf{F}}}_{N_{\rm{RF}}})\right]$;
        \STATE Concatenate TTD-based analog precoders ${\bf{F}}_{{\rm{A}},m}^{\rm{TTD}} = {\rm{blkdiag}}\left(\left[{\widetilde{\bf{p}}}_{1,m},\cdots,{\widetilde{\bf{p}}}_{N_{\rm{RF}},m}\right]\right)$;
        \STATE Obtain the equivalent channel ${\bf{H}}_{{\rm{eq}},m} = {\bf{H}}_{m} {\bf{F}}_{\rm{A}}^{\rm{PS}} {\bf{F}}_{{\rm{A}},m}^{\rm{TTD}}$ with ${\bf{H}}_{{\rm{eq}},m} = {\bf{U}}_{{\rm{eq}},m} {\bf{\Sigma}}_{{\rm{eq}},m} {\bf{V}}_{{\rm{eq}},m}^H $;
		\STATE Determine the digital precoder ${\bf{F}}_{{\rm{D}},m} = {\bf{V}}_{{\rm{eq}},m} {\bf{\Lambda}}$;
		\RETURN ${\bf{F}}_{\rm{A}}^{\rm{PS}}$, ${\bf{F}}_{{\rm{A}},m}^{\rm{TTD}}$ and ${\bf{F}}_{{\rm{D}},m}$.
	\end{algorithmic}
\end{algorithm*}

\section{System Performance Analysis}\label{sec: sys per}
In this section, the system performance analysis on the averaged beamforming gain and spectrum efficiency will be provided.
\subsection{Averaged Beamforming Gain Performance}\label{sec: ave se}
In Section~\ref{sec: dpp}, the beamforming gain at different frequencies with DPP architecture is investigated. Nevertheless, the analysis is only restricted to a single frequency, failing to reveal the overall beamforming performance across the whole bandwidth. To this end, we define the averaged beamforming gain under the single-path assumption to quantify how serious the beam defocus effect is with hybrid precoding architectures, which can be expressed as
\begin{equation}
\label{eq: ave spe define}
\begin{aligned}
\kappa^{\rm{PS}} &= \frac{1}{B} \int_{f_c-B/2}^{f_c+B/2} \left|{\bf{a}}_m^H(\phi){\bf{a}}_c(\phi) \right| {\rm{d}}f_m \\
& \mathop{\approx}\limits^{(a)} \frac{1}{B} \int_{-B/2}^{B/2} \left|J_0\left(\frac{2\pi R}{c}f_m'\right)\right| {\rm{d}}f_m',
\end{aligned}
\end{equation}
where approximation (a) is derived by substituting the results in~\eqref{eq: gain wideband} and defining $f_m' = f_m-f_c$. Due to the difficulty in directly addressing the integral of the absolute Bessel function, we instead analyze its upper and lower bound. According to the Cauchy-schwartz inequality, the upper bound of~\eqref{eq: ave spe define} could be expressed as
\begin{equation}
\label{eq: ave spe upper}
\begin{aligned}
\kappa^{\rm{PS}} \leq \kappa_{\rm{U}}^{\rm{PS}} &= \frac{1}{B} \sqrt{B \int_{-B/2}^{B/2} J_0^2\left(\frac{2\pi R}{c}f_m'\right){\rm{d}}f_m'}\\
& = \sqrt{_2F_3\left(\frac{1}{2},\frac{1}{2};1,\frac{3}{2},\frac{3}{2};-\frac{b_{\rm{PS}}^2}{4}\right)},
\end{aligned}
\end{equation}
where $b_{\rm{PS}}=\frac{\pi BR}{c}$. The proof is similar to the process in {\bf{Appendix}~\ref{app: coro1}} and thus is omitted in this paper. Then, a simple lower bound could be obtained by integration without the absolute operator, which can be expressed as
\begin{equation}
\label{eq: ave spe lower}
\begin{aligned}
\kappa^{\rm{PS}} \geq \kappa_{\rm{L}}^{\rm{PS}} &= \frac{1}{B} \int_{-B/2}^{B/2} J_0\left(\frac{2\pi R}{c}f_m'\right) {\rm{d}}f_m'\\
& =~_1F_2\left(\frac{1}{2};1,\frac{3}{2};-\frac{b_{\rm{PS}}^2}{4}\right),
\end{aligned}
\end{equation}
where equation (b) could be obtained following equation~\eqref{eq: app3 5} in {\bf{Appendix}~\ref{app: coro1}}.
\par For the employed DPP architecture, the averaged beamforming gain can be reformulated as
\begin{equation}
\label{eq: ave spe TTD}
\begin{aligned}
\kappa^{\rm{TTD}} &= \frac{1}{B} \int_{-B/2}^{B/2} \left|_1F_2\left(\frac{1}{2};1,\frac{3}{2};-\frac{\pi^4 R^2}{c^2K^2 f_m'^2}\right) \right| {\rm{d}}f_m' \\
& =~_2F_3\left(\frac{1}{2},\frac{1}{2};1,\frac{3}{2},\frac{3}{2};-\frac{b_{\rm{TTD}}^2}{4}\right),
\end{aligned}
\end{equation}
where $b_{\rm{TTD}} =\frac{\pi^2 BR}{cK}$. Therefore, the beamforming gain improvement employing DPP is at least
\begin{equation}
\label{eq: ave spe improve}
\begin{aligned}
\Delta\kappa &= \frac{\kappa^{\rm{TTD}}}{\kappa_{\rm{U}}^{\rm{PS}}} = \frac{f(b_{\rm{TTD}})}{\sqrt{f(b_{\rm{PS}})}},
\end{aligned}
\end{equation}
where $f(x)$ is defined as $f(x) = _2F_3\left(\frac{1}{2},\frac{1}{2};1,\frac{3}{2},\frac{3}{2};-\frac{x^2}{4}\right)$. An illustration of $f(x)$ can be found as the blue line in Fig.~\ref{img: hyper function}. It can be seen that due to the monotone decreasing property of $f(x)$, a larger $K$ will contribute to an improved beamforming gain. 
% An example of $\Delta\kappa$ related to different settings of $B$ and $K$ will be provided in Section~\ref{sec: sim}.

\subsection{Spectrum Efficiency Analysis}\label{sec: sum rate}
\par In addition to the analysis with the single-path assumption, in this subsection the spectrum efficiency analysis is provided under a more general multi-path scenario. 
\par The spectrum efficiency at the $m^{\rm{th}}$ subcarrier can be expressed as
\begin{equation}
\label{eq: se multi}
\begin{aligned}
R_m = \log_2\left( \left|{\bf{I}} + \frac{\rho}{N_s\sigma_n^2} {\bf{H}}_m^H {\bf{F}}_{\rm{A}} {\bf{F}}_{{\rm{D}},m} {\bf{F}}_{{\rm{D}},m}^H {\bf{F}}_{\rm{A}}^H {\bf{H}}_m \right| \right),
\end{aligned}
\end{equation}
where ${\bf{F}}_{\rm{A}}$ can be further decomposed as ${\bf{F}}_{\rm{A}} = {\bf{F}}_{\rm{A}}^{\rm{PS}} {\bf{F}}_{{\rm{A}},m}^{\rm{TTD}}$. According to~\cite{Ayach'14'j}, when the number of RF chains is set to exceed the number of resolvable paths, the spatial multiplexing gain could be fully harvested. Then, by assuming SVD of the channel ${\bf{H}}_m^H = {\bf{U}}_m {\bf{\Sigma}}_m {\bf{V}}_m^H$ and extracting the significant sub-channels as ${\widetilde{{\bf{\Sigma}}}} = [{\bf{\Sigma}}]_{1:N_{\rm{s}},1:N_{\rm{s}}}$ and ${\widetilde{{\bf{V}}}}_m = [{\widetilde{{\bf{V}}}}_m]_{:,1:N_{\rm{s}}}$, the spectrum efficiency can be reformulated into 
\begin{equation}
\label{eq: se multi 2}
\begin{aligned}
R_m &= \log_2\left(\left|{\bf{I}} + \frac{\rho}{N_s\sigma_n^2} {\widetilde{{\bf{\Sigma}}}}_m {\widetilde{{\bf{V}}}}_m^H {\bf{F}}_{\rm{A}} {\bf{F}}_{{\rm{D}},m} {\bf{F}}_{{\rm{D}},m}^H {\bf{F}}_{\rm{A}}^H {\widetilde{{\bf{V}}}}_m {\widetilde{{\bf{\Sigma}}}}_m^H \right| \right)\\
& = \log_2\left(\left|{\bf{I}} + \frac{\rho}{N_s\sigma_n^2} {\widetilde{{\bf{\Sigma}}}}_m^2 {\bf{V}}_{m{\rm{,eq}}}^H {\bf{V}}_{m{\rm{,eq}}} \right| \right),
\end{aligned}
\end{equation}
where ${\bf{V}}_{m{\rm{,eq}}} = {\bf{F}}_{{\rm{D}},m}^H {\bf{F}}_{{\rm{A}},m}^{{\rm{TTD}}H} {\bf{F}}_{\rm{A}}^{{\rm{PS}}H}  {\widetilde{{\bf{V}}}}_m$. Adopting the linear transformation in~\cite{Ayach'14'j} which reformulates the unitary matrix ${\widetilde{{\bf{V}}}}_m$ with an orthogonal list of beam steering vectors, expressed as\footnote{Although this conclusion is originally derived in ULA systems, the asymptotic orthogonality of beam steering vectors could be extended into UCA systems, see~\cite{Wu'23'arxiv}.}
\begin{equation}
\label{eq: se multi 3}
\begin{aligned}
{\widetilde{{\bf{V}}}}_m \approx {\bf{A}}_{{\rm{t}},m} {\bf{F}}_{\rm{D,m}}^{\rm{opt}},
\end{aligned}
\end{equation}
where ${\bf{A}}_{{\rm{t}},m} = [{\bf{a}}_m(\phi_1),\cdots,{\bf{a}}_m(\phi_{N_{\rm{s}}})]$ and ${\bf{F}}_{\rm{D,m}}^{\rm{opt}}$ denote the optimal digital precoder which could maximize the spectrum efficiency~\cite{Ayach'14'j}. When the number of antennas tends to infinity, the steering vectors could form an orthogonal basis, revealing that ${\bf{A}}_{{\rm{t}},m}$ is a unitary matrix. Since ${\bf{A}}_{{\rm{t}},m}$ and ${\widetilde{{\bf{V}}}}_m$ are both unitary, when ${\bf{F}}_{\rm{A}}^{\rm{PS}} {\bf{F}}_{{\rm{A}},m}^{\rm{TTD}} = {\bf{A}}_{{\rm{t}},m}$ and ${\bf{F}}_{{\rm{D}},m} = {\bf{F}}_{\rm{D,m}}^{\rm{opt}}$ are satisfied, the maximum spectrum efficiency could be obtained as
\begin{equation}
\label{eq: se multi 4}
\begin{aligned}
R_m^{\rm{opt}} &= \log_2\left(\left|{\bf{I}} + \frac{\rho}{N_s\sigma_n^2} {\widetilde{{\bf{\Sigma}}}}_m^2 \right| \right).
% & = \sum_{i=1}^{N_{\rm{s}}} \log_2\left(1 + \frac{\rho}{N_s\sigma_n^2} \lambda_{m,i}^2 \right),
\end{aligned}
\end{equation}
% where $\lambda_{m,i}$ denotes the $i^{\rm{th}}$ largest element in ${\widetilde{{\bf{\Sigma}}}}_m$.
However, due to the beam defocus effect in UCA systems, ${\bf{F}}_{\rm{A}}^{\rm{PS}} {\bf{F}}_{{\rm{A}},m}^{\rm{TTD}} = {\bf{A}}_{{\rm{t}},m}$ could not be perfectly obtained. Therefore, the key factor influencing the spectrum efficiency lies in the analog beamforming gain, which can be expressed as
\begin{equation}
\label{eq: se multi 5}
\begin{aligned}
&~~~~\left({\bf{F}}_{{\rm{A}},m}^{{\rm{TTD}}H} {\bf{F}}_{\rm{A}}^{{\rm{PS}}H} {\bf{A}}_{{\rm{t}},m}\right)^H {\bf{F}}_{{\rm{A}},m}^{{\rm{TTD}}H} {\bf{F}}_{\rm{A}}^{{\rm{PS}}H} {\bf{A}}_{{\rm{t}},m} \\
&= {\rm{blkdiag}}\left(\left[G_m^2({\bf{b}}_{m,1} ,\phi_1),\cdots, G_m^2({\bf{b}}_{m,{N_{\rm{s}}}} ,\phi_{N_{\rm{s}}})\right]\right).
\end{aligned}
\end{equation}
In addition, note that the beamforming gain $G_m({\bf{b}}_{m,i} ,\phi_i)$ is independent of the angle. Therefore, the spectrum efficiency can be further expressed as 
\begin{equation}
\label{eq: se multi 6}
\begin{aligned}
R_m & \mathop{=} \limits^{(a)} \log_2\left( \left|{\bf{I}} + \frac{\rho}{N_s\sigma_n^2} G_m^2({\bf{b}}_{m,l} ,\phi_l) {\bf{\Sigma}}_m^2 \right| \right)\\
& \mathop{=} \limits^{(b)} \log_2\left( \left|{\bf{I}} + \frac{\rho}{N_s\sigma_n^2} ~_1F_2\left(\frac{1}{2};1,\frac{3}{2};-\frac{a^2}{4}\right)^2   {\bf{\Sigma}}_m^2 \right| \right),
\end{aligned}
\end{equation}
where equation (a) is obtained according to the independence of $G_m$ and $\phi_l$ and approximation (b) is derived from the conclusion in {\bf{Corollary}~\ref{coro1}}. The variable $a = \frac{2\pi^2 R}{cK}(f_m-f_c)$ is only determined by the system settings. With a larger $K$, the spectrum efficiency is expected to be enhanced compared with the classical hybrid precoding architecture.

\section{Simulation Results}\label{sec: sim}
\par In this section, simulation results are provided to validate the effectiveness of our theoretical analysis in previous sections. We consider a mmWave wideband communication system with the central frequency $f_c = 30$ GHz and bandwidth $B = 3$ GHz. A $256$-element UCA is equipped at BS to serve a single user equipped with a 4-element ULA. 
\par According to the analysis in {\bf{Corollary}~\ref{coro2}}, if we aim to ensure a beamforming loss less than $\Delta = 40\%$, the minimum required number of TTD units is about $8.2$. Since the number of TTD units has to be an integer, we employ $8$ TTDs in the simulation. The beamforming gain with and without TTD units is plotted in Fig.~\ref{img: beam dpp}. It can be seen from Fig.~\ref{img: beam dpp} that the introduction of TTD has significantly improved the beamforming performance compared with the classical PS-based hybrid precoding architectures. The beamforming gain has approximately exceeded the predetermined beamforming gain threshold of 0.6 over the whole bandwidth, indicating the effectiveness of the results in {\bf{Corollary}~2}. In addition, the green dashed line has perfectly covered the blue and black lines, revealing that the approximation in~\eqref{eq: lemma3 eq2} and~\eqref{eq: coro1 eq1} have achieved high accuracy.
\begin{figure}[!t]
    \centering
    \setlength{\abovecaptionskip}{0.cm}
    \includegraphics[width=3in]{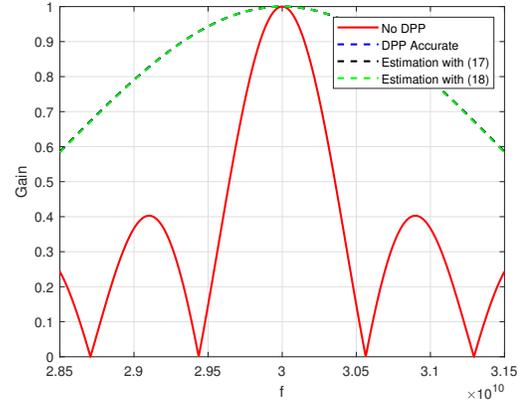}
    \caption{Illustration of the beamforming gain with DPP architectures.}
    \label{img: beam dpp}
\end{figure}
\par To show the superiority of DPP architecture in improving the averaged beamforming gain, the comparisons on averaged beamforming gain between different architectures are plotted in Fig.~\ref{img: averaged beamforming}. In the simulation, $K=8$ is assumed and other systems settings remain the same as in Fig.~\ref{img: beam dpp}. It shows that the DPP architecture obtains the same performance as classical hybrid precoding architectures with a narrow bandwidth. As the bandwidth scales up, the employed DPP architecture always outperforms the classical hybrid precoding scheme. For a bandwidth $B=2$ GHz, the DPP architecture could improve the averaged beamforming gain by about $95\%$. The upper bound and lower bound of the averaged beamforming gain are verified to be effective with the simulation results.
\begin{figure}[!t]
    \centering
    \setlength{\abovecaptionskip}{0.cm}
    \includegraphics[width=3in]{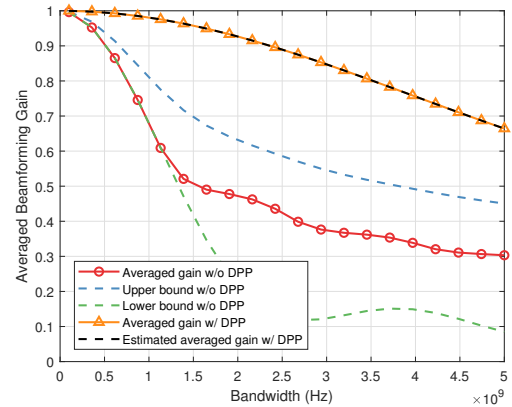}
    \caption{Comparison on the averaged beamforming gain.}
    \label{img: averaged beamforming}
\end{figure}
\par Then, the comparison of the spectrum efficiency over different SNRs is plotted in Fig.~\ref{img: spe eff SNR}. The baselines include the optimal fully-digital precoding, wideband optimization-based method~\cite{Heath'17'twc}, and spatially sparse precoding~\cite{Ayach'14'j}. It can be seen that the DPP architecture outperforms the optimization-based method and spatially sparse precoding at different SNRs. The reason lies in that the analog precoder in the optimization-baed method is only designed to achieve a balanced beamforming gain across the bandwidth, which could not obtain an ideal beamforming gain at different frequencies. Instead, the DPP architecture is able to generate frequency-dependent phase shifts with the aid of TTD, which is expected to obtain an ideal beamforming gain at any frequency. Compared with the optimization-based method in~\cite{Heath'17'twc}, the performance degradation of spatially sparse precoding comes from the non-orthogonality of selected steering vectors, since there are no common methods to form an orthogonal basis consisting of different steering vectors in UCA. 

\begin{figure}[!t]
    \centering
    \setlength{\abovecaptionskip}{0.cm}
    \includegraphics[width=3in]{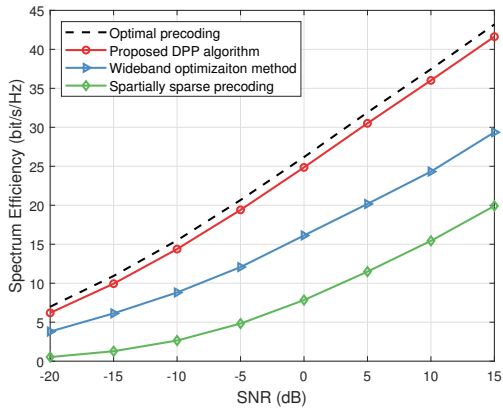}
    \caption{Comparison on the spectrum efficiency over different SNRs.}
    \label{img: spe eff SNR}
\end{figure}

\par To demonstrate the relationship between the spectrum efficiency performance and the number of TTDs employed, the system performance versus the number of TTD units is plotted in Fig.~\ref{img: spe eff TTD}. It can be seen that, if only one TTD is employed, the performance of the DPP architecture is comparable to classical hybrid precoding architectures. The performance of DPP improves as the number of TTDs scales up. When the number of TTDs satisfies $K \geq 8$, the spectrum efficiency of DPP has exceeded $90\%$ of the optimal spectrum efficiency. 

\begin{figure}[!t]
    \centering
    \setlength{\abovecaptionskip}{0.cm}
    \includegraphics[width=3in]{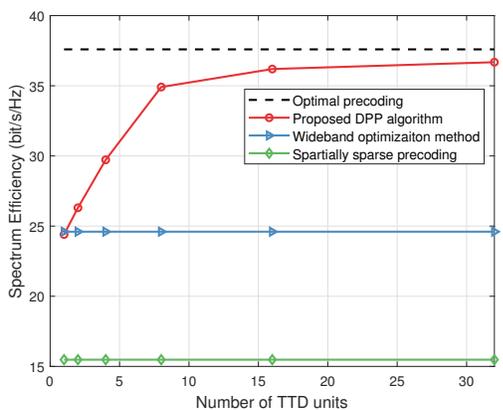}
    \caption{Comparison on the spectrum efficiency over different number of TTD units.}
    \label{img: spe eff TTD}
\end{figure}

\par To show the influence of the bandwidth, the spectrum efficiency performance versus the bandwidth from $100$ MHz to $5$ GHz is plotted in Fig.~\ref{img: spe eff bw}. In this simulation, the number of TTD units is set to $K=16$. With a narrow bandwidth, the optimization-based method could achieve near-optimal performance. Nevertheless, the performance of the optimization-based method will degrade as the bandwidth increases, indicating the influence of the spatial wideband effect on the spectrum efficiency in classical hybrid precoding architectures. On the contrary, the DPP method could always obtain a relatively stable performance for different bandwidths. The performance of DPP architecture begins to decrease when bandwidth exceeds $4$ GHz, which indicates that more TTD units are required.

\begin{figure}[!t]
    \centering
    \setlength{\abovecaptionskip}{0.cm}
    \includegraphics[width=3in]{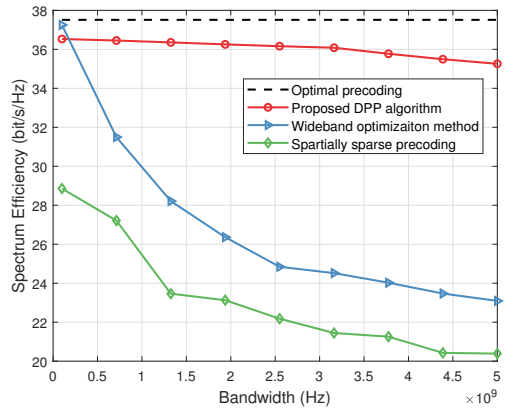}
    \caption{Comparison on the spectrum efficiency over different bandwidths.}
    \label{img: spe eff bw}
\end{figure}

\section{Conclusion}\label{sec: con}
In this paper, the mechanism of the beam defocus effect in UCA systems is investigated for the first time. The dispersed beam pattern in the angular and frequency domain is characterized, revealing that the hybrid precoding architecture may face significant beamforming loss in UCA wideband communications. To retrieve an ideal beamforming gain, the DPP architecture and corresponding precoding algorithm are introduced to mitigate the beam defocus effect in UCA systems. Theoretical analysis and simulation results are provided to verify the effectiveness of the proposed method. The influence of the finite resolution requirement of TTD in UCA wideband communications is left for future research.

\appendices
\section{Proof of Lemma~\ref{lemma2}}\label{app: lemma2}
The proof of {\bf{Lemma 2}} also utilizes the generating function of Bessel functions. First, the equation~\eqref{eq: gain wideband angular} could be written as
\begin{equation}
\label{eq: app1 1}
\begin{aligned}
G_m({\bf{a}}_c(\phi_0),\phi) &= \left| {\bf{a}}_m^H(\phi) {\bf{a}}_c(\phi_0) \right| \\
& = \left|\frac{1}{N} \sum_{n=0}^{N-1} e^{j \left[\eta_m\cos(\phi-\psi_n)-\eta_c\cos(\phi_0-\psi_n)\right]} \right|.
\end{aligned}
\end{equation}
By substituting the generating function of Bessel functions in~\eqref{eq: lemma1 eq1}, we can obtain
\begin{equation}
\label{eq: app1 3}
\begin{aligned}
G_m({\bf{a}}_c(\phi_0),\phi) &= \frac{1}{N} \left| \sum_{n=0}^{N-1} \left[ \sum_{s_1=-\infty}^{+\infty} j^{s_1} J_{s_1}(\eta_m)e^{js_1(\phi-\psi_n)} \right] \right.\\
& \left. \times \left[\sum_{s_2=-\infty}^{+\infty} j^{s_2} J_{s_2}(-\eta_c)e^{js_2(\phi_0-\psi_n)} \right] \right| \\
& \mathop{=}\limits^{(a)} \frac{1}{N} \left| \sum_{s_1=-\infty}^{+\infty} \sum_{s_2=-\infty}^{+\infty} j^{s_1+s_2} J_{s_1}(\eta_m)J_{s_2}(-\eta_c) \right. \\
& ~~~~ \left. \times e^{j(s_1\phi+s_2\phi_0)} \sum_{n=0}^{N-1} e^{-j(s_1+s_2)\psi_n} \right|,
\end{aligned}
\end{equation}
where equation (a) is obtained by exchanging the order of the three summations. When replacing $\psi_n$ with $\frac{2 \pi n}{N}$, the last summation over $n$ could be expressed as the piecewise function similar to~\eqref{eq: lemma1 eq4} as
\begin{equation}
\label{eq: app1 4}
\sum_{n=0}^{N-1} e^{-j(s_1+s_2)\psi_n}=\left\{
\begin{aligned}
N, \quad &s_1+s_2 = N \cdot t, t \in \mathbb{Z} \\
0, \quad &s_1+s_2 \neq N \cdot t, t \in \mathbb{Z}.\\
\end{aligned}
\right.
\end{equation}
Like the proof process of {\bf{Lemma}~1}, the conditions of $s_1+s_2 \neq 0$ could be omitted in the summation when assuming a large $N$. Substituting $s_1+s_2=0$, the beamforming gain could be simplified as
\begin{equation}
\label{eq: app1 5}
\begin{aligned}
G_m({\bf{a}}_c(\phi_0),\phi) &= \left| \sum_{s=-\infty}^{+\infty} J_{s}(\eta_m)J_{-s}(-\eta_c) e^{js(\phi-\phi_0)} \right| \\
& \mathop{=}\limits^{(b)} \left| \sum_{s=-\infty}^{+\infty} J_{s}(\eta_m)J_{s}(\eta_c) e^{js(\phi-\phi_0)} \right|,
\end{aligned}
\end{equation}
where equation (b) is derived with the property $J_s(-x) = (-1)^sJ_s(x)$ and $J_{-s}(x) = (-1)^sJ_s(x)$~\cite{abramowitz1948handbook}. In addition, according to the Addition Theorems of Bessel functions~\cite{abramowitz1948handbook}
\begin{equation}
\label{eq: app1 6}
\begin{aligned}
J_0(r_0) = \sum_{s=-\infty}^{+\infty} J_s(r_1) J_s(r_2) e^{js\theta},
\end{aligned}
\end{equation}
where $r_0 = \sqrt{r_1^2+r_2^2-2r_1r_2\cos\theta}$, the beamforming gain could be rewritten in a very concise expression as
\begin{equation}
\label{eq: app1 7}
\begin{aligned}
G_m({\bf{a}}_c(\phi_0),\phi) = |J_0(\xi)|,
\end{aligned}
\end{equation}
where $\xi = \sqrt{\eta_m^2+\eta_c^2-2\eta_m\eta_c\cos(\phi-\phi_0)}$. This completes the proof.

\section{Proof of Lemma~\ref{lemma3}}\label{app: lemma3}
With the already designed PSs ${\bf{a}}_c^{(t)}(\phi)$, the $n^{\rm{th}}$ element of the product of beamforming vector and steering vector at the $m^{\rm{th}}$ subcarrier is written as
\begin{equation}
\label{eq: app2 1}
\begin{aligned}
[{\bf{a}}_m^{H}(\phi){\bf{a}}_c(\phi)]_n = e^{-j \frac{2\pi R}{c}(f_c-f_m)\cos(\phi-\frac{2\pi n}{N})},
\end{aligned}
\end{equation}
where $n = 0,1,\cdots,N-1$. Noticing that the frequency-dependent phase shifts are not needed at central frequency, ${\bf{p}}_{l,m}$ should be equal to ${\bf{0}}$ when $f_m=f_c$. In addition, when each antenna is equipped with an individual TTD unit, no beam defocus will be generated. Therefore, the corresponding value of ${\bf{p}}_{l,m}$ should be written as
\begin{equation}
\label{eq: app2 2}
\begin{aligned}
{\bf{p}}_{l,m} = [e^{j \frac{2\pi R}{c}(f_c-f_m)\cos\phi}, \cdots, e^{j \frac{2\pi R}{c}(f_c-f_m)\cos(\phi-\frac{2\pi(N-1)}{N})}].
\end{aligned}
\end{equation}
Following this intuition, the $k^{\rm{th}}$ element of ${\bf{p}}_{l,m}$ for connecting to $K$ TTDs should also be expressed as $e^{j \frac{2\pi R}{c}(f_c-f_m)\cos(\phi-\bar{\theta}_k)}$, where $\bar{\theta}_k$ denotes the specific phase shift corresponding to the $k^{\rm{th}}$ subarray.
\par Then, the beamforming gain with DPP architectures could be reformulated as
\begin{equation}
\label{eq: app2 3}
\begin{aligned}
G_m({\bf{b}}_{m,l}, \phi) &= \frac{1}{N} \sum_{k=0}^{K-1} \left[\sum_{i=kP}^{(k+1)P-1} e^{-j\frac{2\pi R}{c}(f_c-f_m)\cos(\phi-\frac{2\pi i}{N})}\right.\\
 &~~~~ \times \left. e^{j\frac{2\pi R}{c}(f_c-f_m)\cos(\phi-\bar{\theta}_k)} \right].
\end{aligned}
\end{equation}
Then, again with the generating function in~\eqref{eq: lemma1 eq1}, we can rewrite the inner summation over $i$ as
\begin{equation}
\label{eq: app2 4}
\begin{aligned}
&~~~\sum_{i=kP}^{(k+1)P-1} e^{-j\frac{2\pi R}{c}(f_c-f_m)\cos(\phi-\frac{2\pi i}{N})}\\
& \mathop{=}\limits^{(a)} \sum_{s_1=-\infty}^{+\infty} j^{s_1} J_{s_1}\left(-\frac{2\pi R}{c}(f_c-f_m)\right)e^{js_1\phi} \sum_{i=kP}^{(k+1)P-1} e^{-js_1 \frac{2\pi i}{N}},
\end{aligned}
\end{equation}
where the equation (a) is obtained by substituting the generating function. Similarly, the outer summation over $k$ could be written as
\begin{equation}
\label{eq: app2 5}
\begin{aligned}
&~~~\sum_{k=0}^{K-1} e^{j\frac{2\pi R}{c}(f_c-f_m)\cos(\phi-\bar{\theta}_k)} \sum_{i=kP}^{(k+1)P-1} e^{-js_1 \frac{2\pi i}{N}}\\
& = \sum_{s_2=-\infty}^{+\infty} j^{s_2} J_{s_2}\left(\frac{2\pi R}{c}(f_c-f_m)\right)e^{js_2\phi} \\
&~~~ \times \sum_{k=0}^{K-1} e^{-js_2 {\bar{\theta}}_k} \sum_{i=kP}^{(k+1)P-1} e^{-js_1 \frac{2\pi i}{N}}\\
& \mathop{=}\limits^{(b)} \sum_{s_2=-\infty}^{+\infty} j^{s_2} J_{s_2}\left(\frac{2\pi R}{c}(f_c-f_m)\right)e^{js_2\phi} \frac{1-e^{-js_1 2\pi/K}}{1-e^{-js_1 2\pi/N}}\\
&~~~ \times e^{-js_2 \frac{\xi}{K}} \sum_{k=0}^{K-1} e^{-j (s_1+s_2) \frac{2\pi k}{K}},
\end{aligned}
\end{equation}
where the equation (b) is derived by assuming ${\bar{\theta}}_k = \frac{2\pi k}{K} + \frac{\xi}{K}$. Then we can adopt a similar process as the proof of {\bf{Lemma}~1.} by setting $s_1+s_2=0$ assuming a relatively large $K$. The beamforming gain could be finally simplified as
\begin{equation}
\label{eq: app2 6}
\begin{aligned}
&~~~~G_m({\bf{b}}_{m,l}, \phi) \\ 
&= \frac{1}{P} \sum_{s=-\infty}^{+\infty} J_s^2\left(\frac{2\pi R}{c}(f_c-f_m)\right) \sum_{i=0}^{P-1} e^{-js \frac{2\pi i - P\xi}{N}}\\
&= \frac{1}{P} \sum_{i=0}^{P-1} \sum_{s=-\infty}^{+\infty} J_s^2\left(\frac{2\pi R}{c}(f_c-f_m)\right) e^{-js \frac{2\pi i - P\xi}{N}}\\
&\mathop{=}\limits^{(c)} \frac{1}{P} \sum_{i=0}^{P-1} J_0(R_i),
\end{aligned}
\end{equation}
where the equation (c) is derived according to the Addition Theorems of Bessel functions~\cite{abramowitz1948handbook} and $R_i = \frac{2\sqrt{2}\pi R}{c}(f_c-f_m)\sqrt{1-\cos\left(\frac{(2i+1)\pi}{N}-\frac{\pi}{K}\right)}$. Finally, the summation could be approximately maximized when $\xi = \pi-\frac{\pi}{P}$, which makes $R_i$ as small as possible. This completes the proof.

\section{Proof of Corollary~\ref{coro1}}\label{app: coro1}
We first denote $\zeta_i = \frac{(2i+1)\pi}{N}-\frac{\pi}{K}$ for $i=0,1,\cdots,P-1$. As a result, $\zeta_i$ satisfies $\zeta_i \in \left[-\frac{\pi}{K}+\frac{\pi}{N},\frac{\pi}{K}-\frac{\pi}{N} \right]$. Assuming a large $K$, $\zeta_i$ can be viewed to around $0$. With the Taylor series expansion $\cos(x) = 1-\frac{1}{2}x^2+\mathcal{O}(x^2)$, $R_i$ in~\eqref{eq: lemma3 eq2} could be simplified as
\begin{equation}
\label{eq: app3 1}
\begin{aligned}
R_i &= \frac{2\sqrt{2}\pi R}{c}(f_c-f_m)\sqrt{1-\cos\zeta_i}\\
&\mathop{\approx}\limits^{(a)} \frac{2\pi R}{c}(f_c-f_m) \zeta_i,
\end{aligned}
\end{equation}
where $i=0,1,\cdots,P-1$. Then, the beamforming gain could be rewritten as
\begin{equation}
\label{eq: app3 2}
\begin{aligned}
G_m({\bf{b}}_{m,l}, \phi) &= \frac{1}{P} \sum_{i=0}^{P-1} J_0\left(\frac{2\pi R}{c}(f_c-f_m) \zeta_i\right)\\
& \mathop{\approx}\limits^{(b)} \frac{N}{2\pi P} \int_{-\pi/K}^{\pi/K} J_0\left(\frac{2\pi R}{c}(f_c-f_m) \zeta\right) {\rm{d}}\zeta\\
& \mathop{=}\limits^{(c)} \frac{N}{\pi P} \int_{0}^{\pi/K} J_0\left(\frac{2\pi R}{c}(f_c-f_m) \zeta\right) {\rm{d}}\zeta,
\end{aligned}
\end{equation}
where approximation (b) is obtained by replacing summation with integral over $\zeta$. Equation (c) is derived according to the parity of $J_0(\cdot)$. By substituting $\zeta'=\frac{2\pi R}{c}(f_c-f_m) \zeta$ and $a = \frac{\pi}{K}\cdot\frac{2\pi R}{c}(f_c-f_m)$, the above equation could be further simplified into
\begin{equation}
\label{eq: app3 3}
\begin{aligned}
G_m({\bf{b}}_{m,l}, \phi) &\approx \frac{1}{a} \int_{0}^{a} J_0\left(\zeta'\right) {\rm{d}}\zeta'\\
&\mathop{=}\limits^{(d)}~_1F_2(\frac{1}{2};1,\frac{3}{2};-\frac{a^2}{4}),
\end{aligned}
\end{equation}
where the proof of equation (d) shall be shown as follows. According to the definition, $\int_0^a J_0(x) {\rm{d}}x$ could be expressed as
\begin{equation}
\label{eq: app3 4}
\begin{aligned}
\frac{1}{a} \int_0^a J_0(x) {\rm{d}}x &= \frac{1}{a} \sum_{n=0}^{\infty} (-1)^n \frac{1}{2^{2n}\Gamma^2(n+1)} \int_0^a x^{2n} {\rm{d}}x\\
& = \sum_{n=0}^{\infty} (-1)^n \frac{1}{(2n+1)2^{2n}\Gamma^2(n+1)} a^{2n}.
\end{aligned}
\end{equation}
While according to the definition of generalized hypergeometric function and the notation $(a)_n = a(a+1)\cdots(a+n-1)$ for $n\geq1$, $_1F_2(\frac{1}{2};1,\frac{3}{2};-\frac{a^2}{4})$ could be expressed 
\begin{equation}
\label{eq: app3 5}
\begin{aligned}
_1F_2(\frac{1}{2};1,\frac{3}{2};-\frac{a^2}{4}) &= \sum_{n=0}^{\infty} \frac{(\frac{1}{2})_n}{(1)_n (\frac{3}{2})_n} \cdot (-1)^n \frac{a^{2n}}{2^{2n} n!}\\
& = \sum_{n=0}^{\infty} (-1)^n \frac{a^{2n}}{2^{2n}(2n+1)n!n!}\\
& = \frac{1}{a} \int_0^a J_0(x) {\rm{d}}x,
\end{aligned}
\end{equation}
which completes the proof.

\footnotesize
\balance 
\bibliographystyle{IEEEtran}
\bibliography{IEEEabrv,reference}
\end{document}